\documentclass[10pt,journal,compsoc]{IEEEtran}

\usepackage{color}
\usepackage{boxedminipage}
\usepackage{amsmath}
\usepackage{fancybox}
\usepackage{multirow}
\usepackage{amsfonts}
\usepackage{mathtools}
\usepackage{amsthm}
\usepackage{mathdesign}
\usepackage{MnSymbol}

\theoremstyle{definition}
\newtheorem{definition}{Definition}[section]
\newtheorem{theorem}{Theorem}[section]

\newcommand{\blue}[1]{\textcolor{black}{#1}}

\ifCLASSOPTIONcompsoc

\usepackage[nocompress]{cite}
\else

\usepackage{cite}
\fi

\ifCLASSINFOpdf

\else

\fi

\begin{document}

	\title{CCA-Secure Key-Aggregate Proxy Re-Encryption for Secure Cloud Storage\thanks{\blue{A partial result of this research was presented in the 2018 IEEE Conference on Dependable and Secure Computing, Kaohsiung, Taiwan, December 10-13, 2018.}}}

	\author{Wei-Hao Chen, Chun-I Fan$^{\ast}$, and Yi-Fan Tseng
		\IEEEcompsocitemizethanks{\IEEEcompsocthanksitem W.-H. Chen is with the Department
			of Computer Science, Purdue University, West Lafayette, IN 47907,
			United States.\protect\\
			E-mail: chen4129@purdue.edu
			
			\IEEEcompsocthanksitem C.-I. Fan is with the Department of Computer Science and Engineering, National Sun Yat-sen University, Kaohsiung, Taiwan, the Information Security Research Center, National Sun Yat-sen University, Kaohsiung, Taiwan, and the Intelligent Electronic Commerce Research Center, National Sun Yat-sen University, Kaohsiung,
			Taiwan.\protect\\
			E-mail: cifan@mail.cse.nsysu.edu.tw ($^{\ast}$The corresponding author)
			 
			\IEEEcompsocthanksitem Y.-F. Tseng is with the Department of Computer Science, National Chengchi University, Taipei 11605, Taiwan.\protect\\
			E-mail: yftseng@cs.nccu.edu.tw
		}
		\thanks{}}

	\markboth{Journal of \LaTeX\ Class Files,~Vol.~14, No.~8, August~2015}%
	{Shell \MakeLowercase{\textit{et al.}}: Bare Demo of IEEEtran.cls for Computer Society Journals}
	\IEEEtitleabstractindextext{
		\begin{abstract}
			The development of cloud services in recent years has mushroomed, 
			for example, Google Drive, Amazon AWS, Microsoft Azure. Merchants 
			can easily use cloud services to open their online shops in a few 
			seconds. Users can easily and quickly connect to the cloud in their 
			own portable devices, and access their personal information effortlessly. 
			Because users store large amounts of data on third-party devices, 
			ensuring data confidentiality, availability and integrity become 
			especially important. Therefore, data protection in cloud storage 
			is the key to the survival of the cloud industry. Fortunately, 
			Proxy Re-Encryption schemes enable users to convert their 
			ciphertext into other's ciphertext by using a re-encryption key. 
			This method gracefully transforms the user's computational cost 
			to the server. In addition, with C-PREs, users can apply their access control right on 
			the encrypted data. Recently, we lowered the key storage cost of 
			C-PREs to constant size and proposed the first Key-Aggregate Proxy 
			Re-Encryption scheme. In this paper, we further prove that our scheme 
			is a CCA-secure Key-Aggregate Proxy Re-Encryption scheme in the 
			adaptive model without using random oracle. Moreover, we also 
			implement and analyze the Key Aggregate PRE application in the 
			real world scenario.
		\end{abstract}

		}

	\maketitle

	\IEEEdisplaynontitleabstractindextext

	\IEEEpeerreviewmaketitle

	\IEEEraisesectionheading{\section{Introduction}\label{sec:introduction}}

	\IEEEPARstart{C}{loud} storage services~\cite{qian2009cloud}, like Microsoft Azure~\cite{jennings2010cloud}, 
	Apple 
	iCloud~\cite{arif2019comparison}, Amazon AWS~\cite{mathew2014overview}, Google Cloud~\cite{bisong2019overview}, have penetrated into 
	every corner of life~\cite{sharma2024cloud}. In today's world, cloud technology~\cite{qian2009cloud} has become an integral part of modern life. Over the past few years, there has been a significant shift in how people store their digital assets. Instead of keeping videos, personal data, photos, music, and other materials on their devices, more individuals are now opting to store them in the cloud \cite{Singh_2017}. As long as the network can be connected, users can easily access data on the cloud. This makes accessing large amounts of data much easier. From a business perspective, merchants can easily deploy their own online stores on the cloud. Mobile app or web developers can also set up servers in the cloud to reduce the cost of storing users' data~\cite{amies2012developing, prodan2009survey}. \blue{More recently, with the rise of AI~\cite{fui2023generative, van2023ai}, cloud storage services have evolved to meet the unique demands of machine learning and data analytics~\cite{zhang2024application}. These platforms now offer specialized storage solutions optimized for AI workloads, capable of handling vast amounts of unstructured data and facilitating rapid data retrieval for model training~\cite{al2023big, kunduru2023recommendations}. These storage innovations enable businesses and researchers to efficiently collect, process, and analyze large volumes of data without the need for extensive on-premises infrastructure~\cite{al2023big, kunduru2023recommendations}.
 }
	
	The growing popularity of cloud storage services means that many people's private files are now being uploaded to remote servers. This raises security concerns~\cite{soveizi2023security, abba2024enabling}, as a successful breach of a cloud system by hackers could potentially compromise thousands of users' data~\cite{sasubilli2021cloud, yurtseven2020review}.
	To prevent user data leakage, we must ensure data authentication, integrity, and confidentiality \cite{Wang_2017, Villari_2013, Subashini_2011}. A common protection method, as shown in Fig \ref{fig:fileenc}, is encrypting data before uploading it to the cloud. This approach prevents hackers from decrypting and accessing users' private information, even if they breach the cloud system~\cite{li2020practical}. However, this method complicates data sharing. For instance, if Alice encrypts and uploads her data to the cloud, and later wants to share it with Bob, she would need to give Bob her private key through a secure channel (Fig \ref{fig:secch}). This practice is impractical as it compromises Alice's private key. Furthermore, building a secure channel is oftentimes impractical in the remote online setting. 
    
	\blue{Fortunately, Proxy Re-Encryption schemes\cite{Ateniese_2009,Canetti_2007,Deng_2008,Green_2009,Libert_TPRE_2008,Libert_uni_CCA_2008,Shao_2009,Tang_2008,Weng_C_PRE_2009,Weng_efficient_C_PRE_2009, agyekum2021proxy, manzoor2021proxy, rawal2021multi, dottling2021universal, lin2020improved, keshta2023blockchain, zhang2023identity, ge2023attribute, lin2024revocable, devaki2022re}} cleverly solve the key sharing problem using a re-encryption key. Here's how it works: When Alice wants to share her encrypted file, she generates a re-encryption key. Later, when Bob requests access to Alice's file, a proxy can use this re-encryption key to transform Alice's ciphertext into Bob's—even if Alice is offline. This elegant solution enables seamless data sharing without compromising security. Fig \ref{fig:pre} illustrates the scenario of proxy re-encryption in cloud environments.

	However, in PRE, access control of the re-encryption key is unrestricted. For instance, if Bob collaborates with the proxy, the proxy can convert all of Alice's ciphertext into Bob's using the re-encryption key. This poses a problem when Alice wishes to share only part of her information with Bob. To address this issue, Wang proposed the Conditional Proxy Re-Encryption scheme (C-PRE) \cite{Weng_C_PRE_2009}. \blue{In C-PREs \cite{Chu_2009,Fang_2009,Liang_2012,Weng_C_PRE_2009,Weng_efficient_C_PRE_2009,Xu_2016, lin2024end, yao2021novel, liang2021attribute, zhang2024conditional, hu2022efficient, tang2023attribute},} the data owner creates an encrypted file and a re-encryption key (conditional key) with a selected conditional value $w$. For example, Alice first chooses a condition value $w$ and uses it to encrypt the file. She then creates a re-encryption key (condition key) using the conditional value $w$ and Bob's public key. This approach limits the decryption rights of Alice's ciphertext to the specified conditional value. Fig \ref{fig:cpre} illustrates this scenario.

	Despite the great advancements in C-PREs, we discovered a limitation: the size of conditional keys grows proportionally with the size of conditions. This poses a significant challenge for devices with limited storage capacity. To address this issue, we introduced the Key-Aggregate Proxy Re-Encryption scheme(KA-PRE)~\cite{Chen_2018}. This scheme not only allows users to implement fine-grained access control over their files but also dramatically reduces the storage cost of re-encryption keys to a constant size. \blue{By the time we prepared this extended version of our paper, there were also other works inspired by our previous KA-PRE scheme~\cite{mahamuni2022blockchain, lin2022improved, pareek2023efficient, wu2024distributed, devaki2022re, fan2020key}. The most relevant works include those by Pareek et al.\cite{pareek2021kapre} and Fan et al.\cite{fan2020key}. Pareek et al.~\cite{pareek2021kapre} demonstrated that their Key-Aggregate Proxy Re-Encryption (KA-PRE) scheme achieves Chosen-Ciphertext Attack (CCA) security under the random oracle model. Fan et al.~\cite{fan2020key}, on the other hand, extended our original scheme to support dynamic conditions. This work builds upon these contributions, advancing the state of the art by providing a rigorous security proof that achieves CCA security without relying on the random oracle model, thus offering stronger security guarantees in standard model.}

	\begin{figure}[h]
		\centering
		\includegraphics[width=0.5\textwidth]{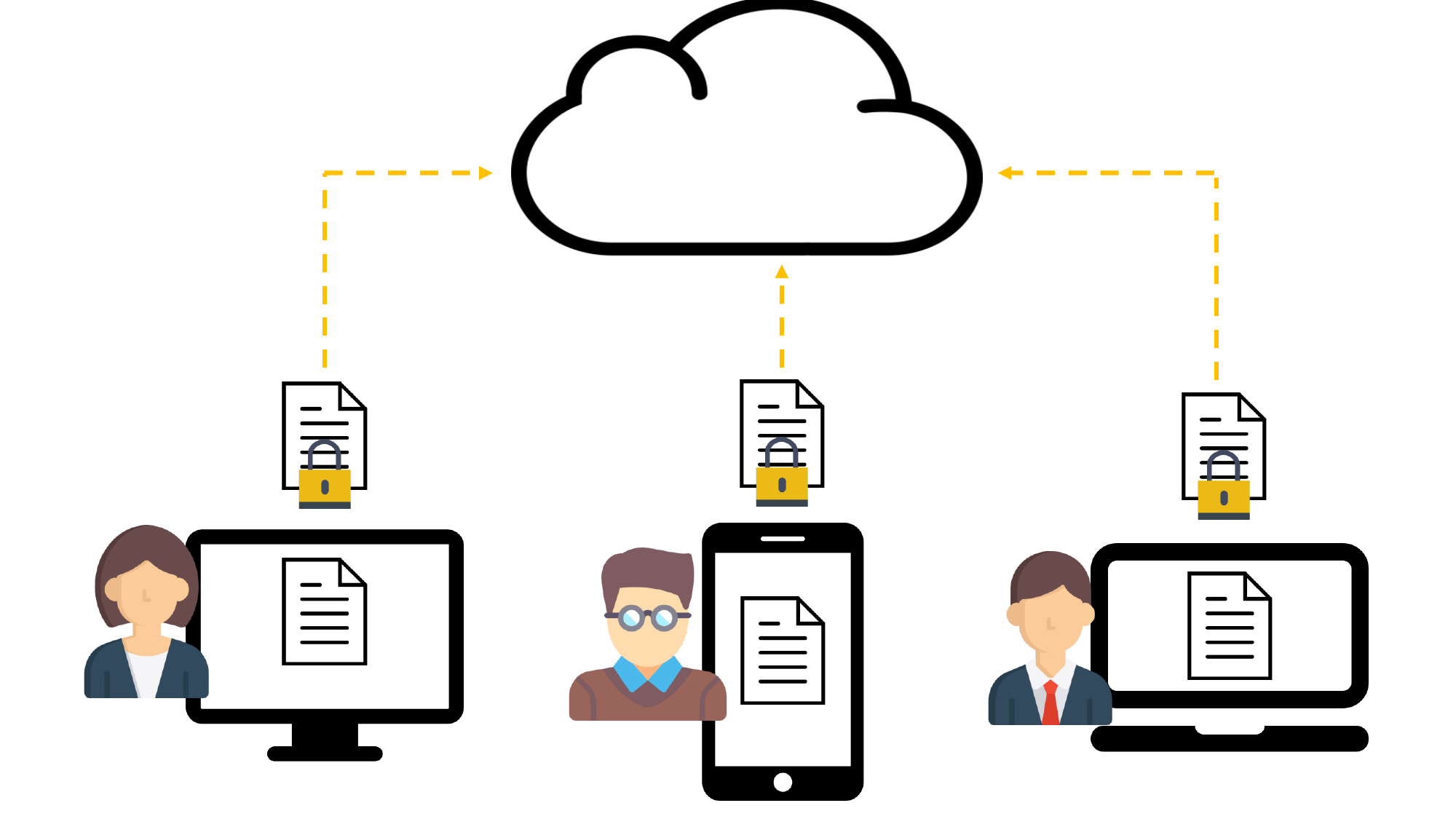}	
		\caption{Users encrypt their files and upload them to the cloud.}
		\label{fig:fileenc}
	\end{figure}
	
	\begin{figure}[h]
		\centering
		\includegraphics[width=0.5\textwidth]{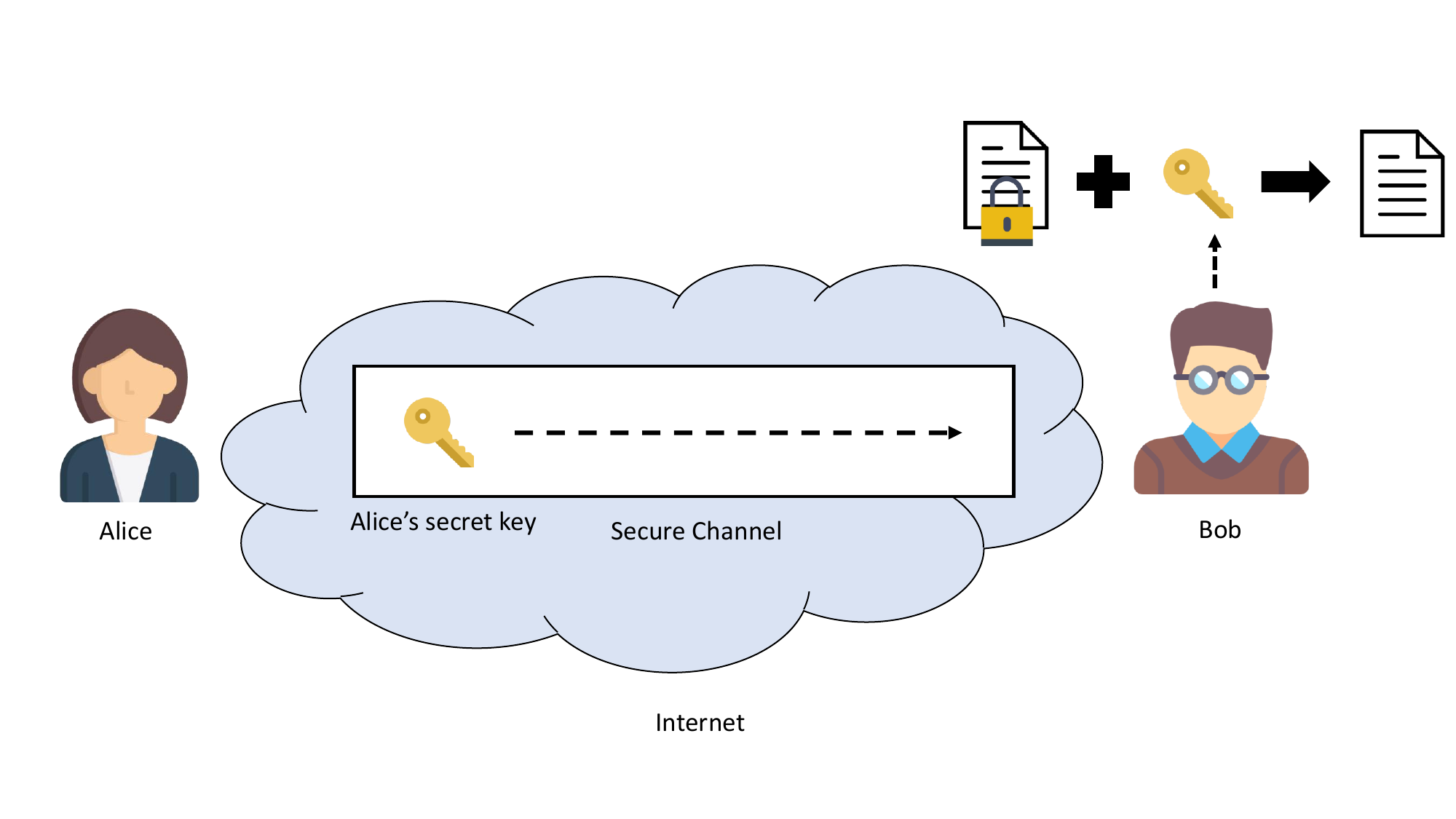}	
		\caption{Alice shares her private key with Bob through a secure channel.}
		\label{fig:secch}
	\end{figure}
	
	\begin{figure}[h]
		\centering
		\includegraphics[width=0.5\textwidth]{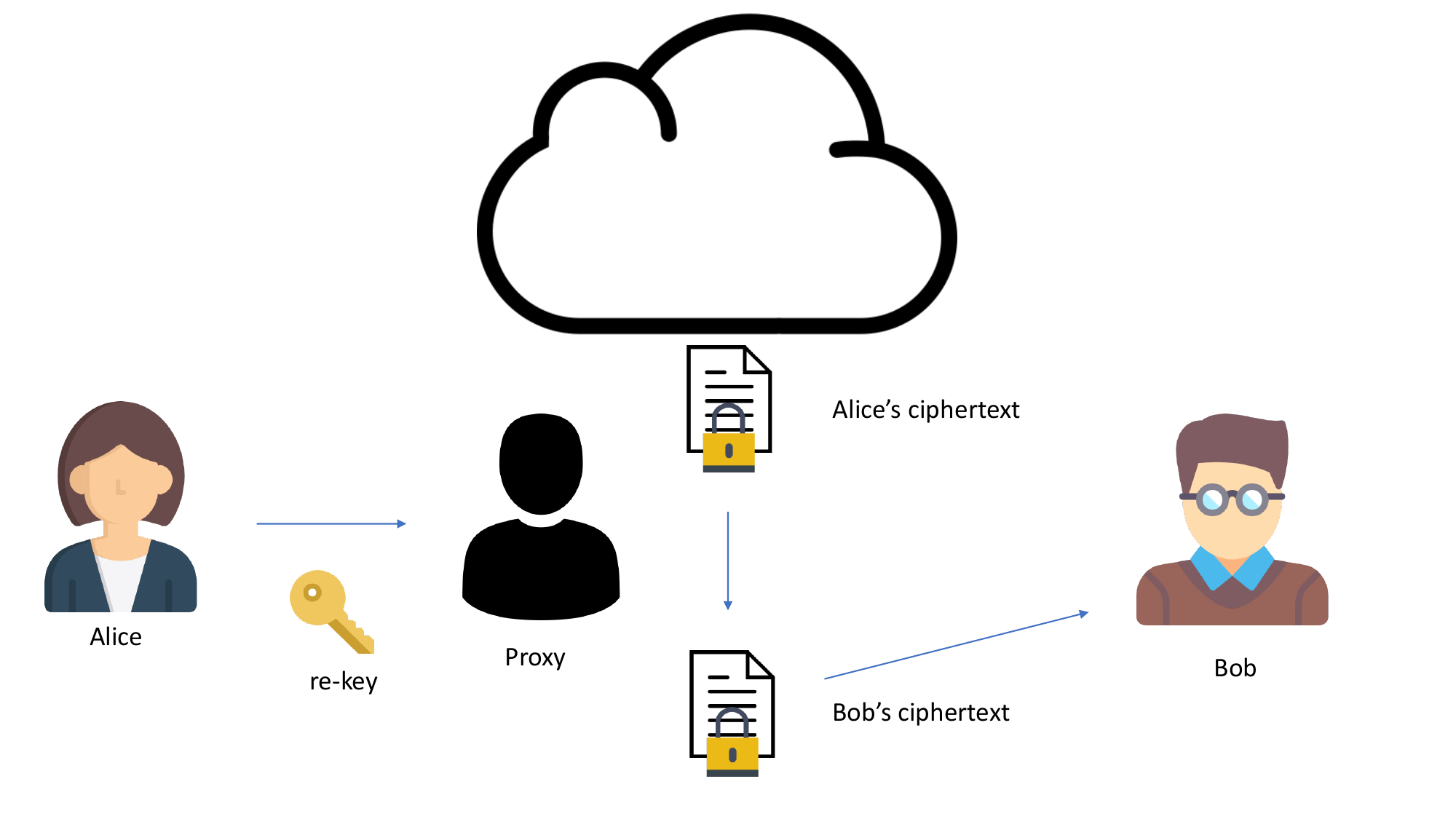}	
		\caption{The cloud manager, the proxy, converts Alice's ciphertext into Bob's via the re-encryption key.}
		\label{fig:pre}
	\end{figure}
	
	\begin{figure}[h]
		\centering
		\includegraphics[width=0.5\textwidth]{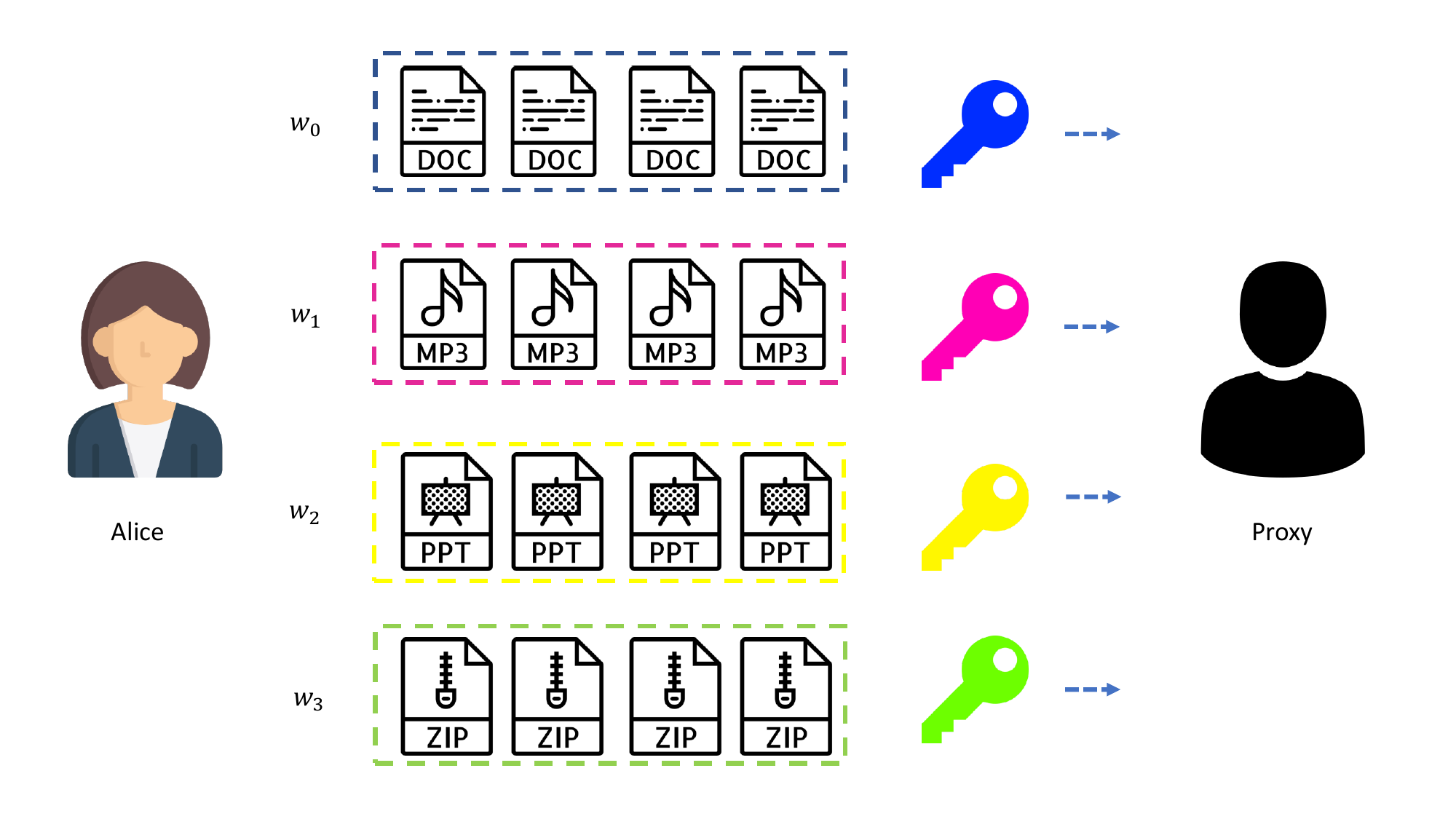}	
		\caption{Alice uses four different conditional keys to employ different access control right on the encrypted files.}
		\label{fig:cpre}
	\end{figure}
	
	\subsection{Contribution}
	\label{sec:contribution}
	In this paper, we further prove that our KAPRE scheme is CCA 
	secure in the adaptive model without using random oracle. 
	In addition, we also adopt the PBC library to implement the 
	proposed scheme in the real world scenario.

	\section{Preliminaries}
	\label{ch:preliminaries}
	Some basic preliminaries and the security definition are 
	reviewed and fomalized in this section.
	\subsection{Bilinear Mapping}
	\label{sec:Bilinear}
	\begin{definition}
		Groups ($\mathbb{G}$, $\mathbb{G}_T$) are two multiplicative cyclic groups of prime order $p$. \\There is a bilinear mapping $e:\mathbb{G} \times \mathbb{G} \to \mathbb{G}_T$ with the following properties.
		\begin{itemize}
			\item{\bf Bilinearity:} $e(g^a,h^b)=e(g,h)^{ab}$ for any $g, h \in \mathbb{G}$ and $a, b \in \mathbb{Z}_p$.
			\item{\bf Non-Degeneracy:} Whenever $g,h \neq 1_{\mathbb{G}}$ , $e(g,h) \neq 1$.
			\item{\bf Computability:} There exists an  algorithm  that can efficiently compute $e(g,h)$ for any $g, h \in \mathbb{G}$.
		\end{itemize}
	\end{definition}
	
	\subsection{3-weak Decisional Bilinear Diffie-Hellman Inversion}
	\label{sec:3-wDBDHI}
	\begin{definition}
		(3-weak Decisional Bilinear Diffie-Hellman Inversion \cite{Libert_uni_CCA_2008}).\\ 
		We define the 3-weak Decisional Bilinear Diffie-Hellman Inversion (3-wDBDHI) problem as follows. Select two groups $\mathbb{G}$ and $\mathbb{G}_T$  with prime order $p$ and choose $(g,g^a,g^{(a^2)},g^{(a^3)},g^b,Q)$ for random $g\in \mathbb{G}$, $a,b \in_R \mathbb{Z}_p$, and $Q \in_R \mathbb{G}_T$, and then distinguish $e(g,g)^{b/a}$ from $Q  \in \mathbb{G}_T$. In \cite{Libert_uni_CCA_2008}, the authors pointed out that 3-weak Decisional Bilinear Diffie-Hellman Inversion problem is equivalent to deciding whether Q is equal to $e(g,g)^{b/a^2}$ given $(g,g^{1/a},g^a,g^{(a^2)},g^b,Q)$. \\ An algorithm $\mathcal{A}$ $(t,\epsilon)$ has advantage at least $\epsilon$ in solving 3-wDBDHI problem if it runs in time $t$ where
		$$
		\begin{array}{c}
		| \Pr [\mathcal{A}(g,g^{1/a},g^a,g^{(a^2)},g^b,e(g,g)^{b/a^2})=1] - \\
		\Pr [\mathcal{A}(g,g^{1/a},g^a,g^{(a^2)},g^b,Q)=1 ] | \geq\epsilon.
		\end{array}
		$$
		\\
		We say that the $(\epsilon, t)$-3-wDBDHI assumption holds in a group $\mathbb{G}$ if no polynomial $t$-time algorithm can solve the 3-wDBDHI problem with non-negligible advantage $\epsilon$.
	\end{definition}
	\subsection{Model and Security Notions of Unidirectional Key-Aggregate PRE}
	\label{sec:Model and Security Notions of Unidirectional Key-Aggregate PRE}
	The section presents the proxy re-encryption model \cite{Libert_uni_CCA_2008,Weng_2010} with some modifications. We only focus on defining a unidirectional, single-use proxy re-encryption scheme. One can simply modify this definition into a bidirectional or multi-use scheme but it is not in the scope of our discussion.
	\subsubsection{Unidirectional Key-Aggregate PRE }
	\label{Unidirectional Key-Aggregate PRE }
	\begin{definition}
		A unidirectional, single-use, proxy re-encryption scheme $\Gamma$ is composed of eight algorithms  $\Gamma$ = (Setup, KeyGen, ReKeyGen, Enc\textsubscript{2}, Enc\textsubscript{1}, ReEnc, Dec\textsubscript{2}, Dec\textsubscript{1}):
		\begin{itemize}
			\item[-] { Setup ($\lambda$)} :\\	The Setup algorithm takes  $\lambda$ as a security parameter and then publishes the public parameters $par$.
			
			\item[-] { KeyGen ($i$)} :\\ On input of user's index $i$, KeyGen generates user $i$'s key pair $(pk_i, sk_i)$ in which $pk_i$ is the public key and $sk_i$ is the secret key.
			
			\item[-] { ReKeyGen ($S,sk_i,pk_j$)} :\\ Given $sk_i$ of user $i$, $pk_j$ of user $j$, ReKeyGen outputs a re-encryption key $rk_{i \xrightarrow[S]{} j}$ where the set $S$ contains the indices of the files that are able to be re-encrypted.	
			
			\item[-] { Enc\textsubscript{2} ($pk_i,m$)} :\\ On input of a message $m$ and $pk_i$ , this algorithm outputs a second level ciphertexts $C$.
			
			\item[-] { Enc\textsubscript{1} ($pk_j,m$)} :\\ On input of a message $m$ and $pk_j$ , this algorithm outputs a first level ciphertexts $C^\prime$.
			
			\item[-] { ReEnc ($S,rk_{i \xrightarrow[S]{} j},C$)} :\\ Given a re-encryption key $rk_{i \xrightarrow[S]{} j}$, ciphertext $C$ for user $i$, and $S$, ReEnc ouputs ciphertext $C^\prime$ for user $j$ if the file type of ciphertext $C$ is in the set $S$.
			
			\item[-] { Dec\textsubscript{2} ($sk_i,C$)} :\\ On input of a secret key $sk_i$ and a second level ciphertext $C$ for user $i$, Dec\textsubscript{2} outputs plaintext $m$ or symbol $\bot$ if the ciphertext is invalid.
			
			\item[-] { Dec\textsubscript{1} ($sk_j,C^\prime$)} :\\ On input of a secret key $sk_j$ and a first level ciphertext $C^\prime$ for user $j$, Dec\textsubscript{1} outputs plaintext $m$ or symbol $\bot$ if the ciphertext is invalid.

		\end{itemize}
		A unidirectional, single-use proxy re-encryption scheme $\Gamma$ is correct if:
		\begin{itemize}
			\item[-]{} For all $(pk,sk) \leftarrow$ { KeyGen($par$)} and $m$, \[\text{Dec}\textsubscript{1}(sk, \text{Enc}\textsubscript{1}(pk,m)) = m.\] \[\text{Dec}\textsubscript{2}(sk, \text{Enc}\textsubscript{2}(pk,m)) = m.\]
			\item[-]{} For all $(pk_i , sk_i), (pk_j , sk_j ), rk_{i \xrightarrow[S]{} j}$, and $m$, \[\text{Dec}\textsubscript{1}(sk_j , ReEnc(S,rk_{i \xrightarrow[S]{} j} , \text{Enc}\textsubscript{2}(pk_i , m))) = m.\]
		\end{itemize} 
		
	\end{definition}
	\subsection{Chosen-Ciphertext Security}
	\label{sec:Chosen-Ciphertext Security}
	
	In this section, we review security notions for unidirectional PRE schemes derived from \cite{Canetti_2007,Libert_uni_CCA_2008,Weng_2010}. We consider the following oracles that model the ability of an adversary $\mathcal{A}$:
	\begin{itemize}
		\item[-] Public key generation oracle $\mathcal{O}_{pk}(i)$:\\ This oracle takes an index $i$ as an input, runs algorithm KeyGen ($par$), and returns $pk_i$ to $\mathcal{A}$. 
		\item[-] Secret key generation oracle $\mathcal{O}_{sk}(pk_i)$:\\ To obtain $(pk_i, sk_i)$, this oracle takes $pk_i$ as an input and returns $sk_i$ with respect to $pk_i$.
		\item[-] Re-encryption key generation oracle  $\mathcal{O}_{rk}(S,pk_i,pk_j)$:\\ Return $rk_{i \xrightarrow[S]{} j}\leftarrow$ ReKeyGen($S,sk_i,pk_j$).
		\item[-] Re-encryption oracle  $\mathcal{O}_{re}(S,pk_i,pk_j,C)$:\\ This oracle returns $C^\prime\leftarrow$ ReEnc($S$, ReKeyGen($S,sk_i,pk_j$),$C$) to $\mathcal{A}$.
		\item[-] First level decryption oracle  $\mathcal{O}_{1d}(pk_j,C^\prime)$:\\ This oracle returns the result of Dec\textsubscript{1}($sk_j,C^\prime$) where $C^\prime$ is a first level ciphertext for user $j$.
		
	\end{itemize}
	
	\subsection{Security of Second Level Ciphertext}
	\label{sec:Security of Second Level Ciphertext}
	Our definition is inspired from \cite{Weng_2010} that uses adaptive corruption model in the security proof. Besides, the security is CCA secure instead of RCCA secure. We denote the following security as IND-2PRE-CCA.
	\begin{definition}
		A polynomial time adversary $\mathcal{A}$'s advantage against the chosen-ciphertext atacks at level 2 (IND-2PRE-CCA) for a single-use unidirectional proxy re-encryption scheme $\Gamma$ is as follows:
		$$ 
		\begin{array}{l}
		Adv^{IND-2PRE-CCA}_{\Gamma,\mathcal{A}}(\lambda)=\\\\
		\left| \Pr \left[ \delta^\prime=\delta \left|\begin{array}{c}
		par \leftarrow \text{Setup}(\lambda);\\(pk_{i^*},(m_0,m_1),st)\leftarrow\\
		\mathcal{A}_{phase1}^{\mathcal{O}_{pk},\mathcal{O}_{sk},\mathcal{O}_{rk},\mathcal{O}_{re},\mathcal{O}_{1d}}(par);\\
		\delta \xleftarrow{\$}\{0,1\};C^* \leftarrow \text{Enc}\textsubscript{2}(pk_{i^*},m_\delta); \\
		\delta^\prime \leftarrow \mathcal{A}_{phase2}^{\mathcal{O}_{pk},\mathcal{O}_{sk},\mathcal{O}_{rk},\mathcal{O}_{re},\mathcal{O}_{1d}}(par,C^*,st);\\
		\end{array}\right.
		\right] - \frac{1}{2} \right|
		\end{array}
		$$
		in which $st$ is the internal state of the adversary $\mathcal{A}$. The following rules should be satisfied: (i). $\mathcal{A}$ cannot ask for $\mathcal{O}_{sk}(pk_{i^*})$; \enspace (ii). $\mathcal{A}$ cannot ask for $\mathcal{O}_{sk}(pk_j)$ and $\mathcal{O}_{rk}(pk_{i^*},pk_j)$ at the same time; \enspace  (iii). $\mathcal{A}$ cannot ask for $\mathcal{O}_{sk}(pk_j)$ and  $\mathcal{O}_{re}(pk_{i^*},pk_j,C^*)$ at the same time; \enspace  (iv). For a first level ciphertext $C^{\prime*}$ generated by $\mathcal{O}_{re}(pk_{i^*},pk_j,C^*)$, $\mathcal{A}$ cannot ask for $\mathcal{O}_{1d}(pk_j,C^{\prime*})$.  We denote the above adversary $\mathcal{A}$ as an IND-2PRE-CCA adversary. We say that a single-use unidirectional proxy re-encryption scheme $\Gamma$ is $(t,q_{pk},q_{sk},q_{rk},q_{re},q_{1d},\epsilon)$-IND-2PRE-CCA secure if there is no $t$-time IND-2PRE-CCA adversary $\mathcal{A}$ such that $ Adv^{IND-2PRE-CCA}_{\Gamma,\mathcal{A}}(\lambda)\geq \epsilon$, where $q_{pk},q_{sk},q_{rk},q_{re},q_{1d}$ are the numbers of queries to oracle $\mathcal{O}_{pk},\mathcal{O}_{sk},\mathcal{O}_{rk},\mathcal{O}_{re},\mathcal{O}_{1d}$ respectively.\\
		\textit{Remark.}\enspace In \cite{Libert_uni_CCA_2008}, the author pointed out that providing $\mathcal{A}$ with a second level decryption oracle is meaningless, since (i). $\mathcal{A}$ is not allowed to issue a second level decryption query to decrypt the challenge ciphertext $C^*$; \enspace (ii). For any other second level ciphertext, $\mathcal{A}$ can first ask $\mathcal{O}_{re}$ to re-encrypt it into a first level ciphertext, and then ask $\mathcal{O}_{1d}$ to decrypt it.
	\end{definition} 
	\subsection{Security of First Level Ciphertext}
	\label{sec:Security of First Level Ciphertext}
	In this definition, we review the security of first level ciphertext in \cite{Weng_2010} where $\mathcal{A}$ is provided with first level challenge ciphertexts in the challenge phase. Note that the first level ciphertext cannot be re-encrypted further. Therefore, $\mathcal{A}$ can obtain any re-encryption key including those from the target public key $pk_{i^*}$ to other public keys which are produced by $\mathcal{O}_{pk}$. Besides, as $\mathcal{A}$ is allowed to obtain any re-encryption keys, the challenger does not have to provide the re-encryption oracle $\mathcal{O}_{re}$ for $\mathcal{A}$. As aforesaid remark, providing $\mathcal{A}$ with a second level decryption oracle is meaningless in the IND-1PRE-CCA definition.
	\begin{definition}
		A polynomial time adversary $\mathcal{A}$'s advantage against the chosen-ciphertext atacks at level 1 (IND-1PRE-CCA) for a single-use unidirectional proxy re-encryption scheme $\Gamma$ is as follows:
		$$
		\begin{array}{l}
		Adv^{IND-1PRE-CCA}_{\Gamma,\mathcal{A}}(\lambda)=\\\\
		\left| \Pr \left[ \delta^\prime=\delta \left|
		\begin{array}{c}
		par \leftarrow \text{Setup}(\lambda);\\(pk_{i^*},(m_0,m_1),st)\leftarrow\\
		\mathcal{A}_{phase1}^{\mathcal{O}_{pk},\mathcal{O}_{sk},\mathcal{O}_{rk},\mathcal{O}_{1d}}(par);\\
		\delta \xleftarrow{\$}\{0,1\};C^{\prime*} \leftarrow \text{Enc}\textsubscript{1}(pk_{i^*},m_\delta); \\
		\delta^\prime \leftarrow \mathcal{A}_{phase2}^{\mathcal{O}_{pk},\mathcal{O}_{sk},\mathcal{O}_{rk},\mathcal{O}_{1d}}(par,C^{\prime*},st);\\
		\end{array}\right.
		\right] - \frac{1}{2} \right|
		\end{array}
		$$
		where $st$ is the internal state information of adversary $\mathcal{A}$. Note that $\mathcal{A}$ cannot issue either $\mathcal{O}_{sk}(pk_{i^*})$ or $\mathcal{O}_{1d}(pk_{i^*},C^{\prime*})$. We denote the above adversary $\mathcal{A}$ as an IND-1PRE-CCA adversary. We say that a single-use unidirectional proxy re-encryption scheme $\Gamma$ is $(t,q_{pk},q_{sk},q_{rk},q_{1d},\epsilon)$-IND-1PRE-CCA secure if there is no $t$-time IND-1PRE-CCA adversary $\mathcal{A}$ such that $ Adv^{IND-1PRE-CCA}_{\Gamma,\mathcal{A}}(\lambda)\geq \epsilon$, where $q_{pk},q_{sk},q_{rk},q_{1d}$ are the numbers of queries to oracle $\mathcal{O}_{pk},\mathcal{O}_{sk},\mathcal{O}_{rk},\mathcal{O}_{1d}$ respectively.\\
	\end{definition}
	\subsection{Master Secret Security}
	\label{sec:Master Secret Security}
	In \cite{Ateniese_2006}, the authors defined master secret security for unidirectional proxy re-encryption schemes. This security notion describes that even if all the delegatees collude together, it is still difficult for them to compute the delegator's secret key. Followings is the security notion of MSS-PRE:
	\begin{definition}
		An adversary $\mathcal{A}$'s advantage against master secret security for a unidirectional proxy re-encryption scheme $\Gamma$ is defined as:
		$$
		\\
		\begin{array}{l}
		Adv_{\Gamma,\mathcal{A}}^{MSS-PRE}(\lambda)=\Pr[sk_{i^*} \text{ is a valid secret key}| par \leftarrow \\ \text{Setup}( \lambda );sk_{i^*} \leftarrow \mathcal{A}^{\mathcal{O}_{pk},\mathcal{O}_{sk},\mathcal{O}_{rk}}(par) ]
		\end{array}
		\\
		$$  
		in which the public key $pk_{i^*}$ is produced by $\mathcal{O}_{pk}$, and $\mathcal{A}$ cannot ask for $\mathcal{O}_{sk}(pk_{i^*})$. The above adversary is an MSS-PRE adversary. We say that a unidirectional proxy re-encryption scheme $\Gamma$ is $(t,q_{pk},q_{sk},q_{rk},\epsilon)$-MSS-PRE secure if there is no $t$-time MSS-PRE adversary $\mathcal{A}$ such that $ Adv^{MSS-PRE}_{\Gamma,\mathcal{A}}(\lambda)\geq \epsilon$, where $q_{pk},q_{sk},q_{rk}$ are the numbers of queries to oracle $\mathcal{O}_{pk},\mathcal{O}_{sk},\mathcal{O}_{rk}$ respectively.
	\end{definition}
	\begin{theorem}
		\label{MSS}
		Suppose that there is an adversary $\mathcal{A}$ who has the ability to break the MSS-PRE security of
		a single-use unidirectional proxy re-encryption scheme $\Gamma$, then we can build another adversary $\mathcal{B}$ who can break IND-1PRE-CCA security.
	\end{theorem}
	\noindent That is, for a single-use unidirectional proxy re-encryption scheme, the master secret security is implied by the first level ciphertext security. The proof for Theorem \ref{MSS} is intuitive, and thus the proof is omitted here.

	\section{Our Construction}
	\label{ch:construction}
	In this section, we propose a key-aggregate proxy re-encryption scheme which achieves fine-grained access control on the shared files. Moreover, the size of a re-encryption key is constant. The scheme is IND-PRE-CCA secure based on the 3-wDBDHI assumption in an adaptive model without using random oracles.
	\subsection{Notations}
	The notations used here are defined in TABLE~\ref{tab:notation2}.\\
	\begin{table}[tbh]
		\centering
		\caption{The Notations}
		\begin{tabular}{|cl|}
			\hline
			Notation & Meaning\\
			\hline
			$\mathbb{G}$							&	a cyclic multiplicative group of prime order $p$\\
			$\mathbb{G}_T$							&	a cyclic multiplicative group of prime order $p$\\
			$e$							&	a bilinear mapping; $e:{\mathbb{G}}\times{\mathbb{G}}\to{\mathbb{G}_T}$\\
			$par$						&   the public parameters \\
			$n$						&	the number of total types\\ 
			\hline
		\end{tabular}
		\label{tab:notation2}
	\end{table}
	\subsection{The Proposed Scheme}
	\label{sec:scheme}
	The proposed scheme, noted by $\Gamma$, consists of eight algorithms, i.e. Setup, KeyGen, ReKeyGen, Enc\textsubscript{2}, Enc\textsubscript{1}, ReEnc, Dec\textsubscript{2}, and Dec\textsubscript{1}, which will be defined in the following subsections.
	\subsubsection{Setup ($\lambda$)}
	\label{subsec:setup}
	This algorithm outputs the parameters $par=\langle p,g,d,u,v,w,e,\mathbb{G},\mathbb{G}_T,Z,H,F,l_1,l \rangle$ upon an input security parameter $\lambda$. Each component of $par$ is set as follows. $\mathbb{G}, \mathbb{G}_T$ are two cyclic multiplicative groups of prime order $p$ where  
	$2^{\lambda}\leq p \leq 2^{\lambda+1}$ and $g,d,u,v,w$ are generators of $\mathbb{G}$. Let $H:\mathbb{G}\times\mathbb{G}\rightarrow\mathbb{Z}_p^*, H_1:\mathbb{G}_T \times \mathbb{G} \rightarrow\mathbb{Z}_p^* $ be two collision resistant hash functions and $Z=e(g,g)$. 
	\subsubsection{KeyGen ($i$)}
	\label{subsec:KeyGen}
	A user $i$'s key pair is of the form 
	$pk_i=(g^{a_{i,1}},g^{a_{i,2}},\Delta_i)$, $sk_i=(a_{i,1},a_{i,2},a_{i,3})$, where $a_{i,1},a_{i,2},a_{i,3} \in_R \mathbb{Z}_p$, $\Delta_i=\langle g_1,g_2,\dots,g_n,g_{n+2},\dots,g_{2n} \rangle$, where $g_\rho = g^{a_{i,3}^\rho}$ for each $\rho \in \{1, ..., 2n\}/\{n+1\}$, and $n$ represents the number of the file types specified by user $i$. 
	\subsubsection{ReKeyGen ($S,sk_i,pk_j$)}
	\label{subsec:ReKeyGen}
	For the set $S$ of user $i$'s file types that are able to be re-encrypted, user $i$ (a delegator) with $sk_i=(a_{i,1},a_{i,2},a_{i,3})$ can delegate decryption right to a delegatee, say user $j$, with $pk_j=(g^{a_{j,1}},g^{a_{j,2}},\Delta_j)$ by computing $rk_{i \xrightarrow[S]{} j}$ as: 
	
	$$
	\begin{array}{l}
	rk_{i \xrightarrow[S]{} j}=((g^{a_{j,1}})^{1/a_{i,1}}),(\prod_{\nu\in S}g_{n+1-\nu})^{a_{i,2}})= \\ (g^{a_{j,1}/ a_{i,1} },\prod_{\nu\in S}g_{n+1-\nu}^{a_{i,2}}).
	\end{array}
	$$
	\\
	Suppose that we have totally 10 different kinds of file types, we have: 
	$$
	n + 1 = 10 + 1 = 11
	$$
	When Alice wants to share file types $S = \left\lbrace 1, 3, 5, 8, 9\right\rbrace $ with Bob, she can generate a re-encryption key $rk_{i \xrightarrow[S]{} j} = \left\langle r_1, r_2 \right\rangle $ as follows:
	$$
	r_1 = \left( g^{a_{j,1}}\right) ^{1/a_{i,1}},  r_2 =\left( g_{10} \cdot g_{8} \cdot g_{6} \cdot g_{3} \cdot g_{2} \right)^{a_{i,2}} 
	$$
	\subsubsection{Enc\textsubscript{2} ($pk_i,m$)}
	\label{subsec:Enc}
	Given $m \in \mathbb{G}$ and user $i$'s public key $pk_i=(g^{a_{i,1}},g^{a_{i,2}},\Delta_i)$, where $m$'s type is  $\rho \in \{1, ..., n\}$, select $t,k,r,\eta \in_R \mathbb{Z}_p $ and compute the ciphertext as:
	$$
	\begin{array}{l}
	C = (k,c_1,c_2,c_3,c_4,c_5,c_6,c_7,c_8,c_9) \\ = (k, d^{r}, g^{a_{i,1}r},g^t,(g^{a_{i,2}}\cdot g_\rho)^t, H_1(K,d^r) \oplus m \oplus H_1\\ (e(g_1,g_n)^t, d^r),
	(u^h\cdot v^k\cdot w)^r,(u^{h^\prime}\cdot v^k\cdot w)^r,g^\eta, H_1(K,g^\eta))
	\end{array}
	$$
	
	\noindent where $h=H(c_1,c_5),h^\prime=H(c_1, H_1(K,c_1)\oplus m), K=Z^r$. Finally, this algorithm outputs the ciphertext: 
	$$
	\begin{array}{l}C=(k,c_1,c_2,c_3,c_4,c_5,c_6,c_7,c_8,c_9)\end{array}
	$$.
	Note that, the validity of a second level ciphertext can be verified by checking whether the following formulas hold:

	\[
	e(c_1,u^h\cdot v^k \cdot w)=e(c_6,d)\label{eq1}\tag{Eq. 1} 
	\]
	\[
	e(c_1,g^{a_1})=e(c_2,d)\label{eq2}\tag{Eq. 2} 
	\]
	\[
	e(c_3,g^{a_{i,2}}\cdot g_\rho)=e(c_4,g)\label{eq3}\tag{Eq. 3} 
	\]
	
	\noindent If true, it is accepted as a valid ciphertext.
	
	\subsubsection{Enc\textsubscript{1} ($pk_j,m$)}
	\label{subsec:Enc}
	Given $m \in \mathbb{G}$ and user $j$'s public key $pk_j=(g^{a_{j,1}},g^{a_{j,2}},\Delta_j)$, select $k,r,\eta \in_R \mathbb{Z}_p $ and compute the ciphertext as:
	$$
	\begin{array}{l}
	C^\prime = (k, c_1^\prime, c_2^\prime, c_3^\prime, c_4^\prime, c_5^\prime, c_6^\prime) = (k, d^{r}, e(g^{a_{j,1}},g)^r, H_1(K,d^r)\oplus \\ m, (u^h\cdot v^k\cdot w)^r,g^\eta, H_1(K,g^\eta))
	\end{array}
	$$
	\noindent where $h=H(c_1^\prime,c_3^\prime),$ and $K=Z^r$. Finally, output the ciphertext $C^\prime=(k,c_1^\prime,c_2^\prime,c_3^\prime,c_4^\prime, c_5^\prime, c_6^\prime)$.\\
	Note that the validity of $k, c_1^\prime, c_3^\prime, c_4^\prime$ can be verified by checking:
	\[
	e(c_1^\prime,u^h\cdot v^k \cdot w)=e(c_4^\prime,d)\label{eq4}\tag{Eq. 4} 
	\]
	The validity of $c_2^\prime$ can be verified as that in \ref{eq6}.  If true, it is accepted as a valid ciphertext.

	\subsubsection{ReEnc ($S,rk_{i \xrightarrow[S]{} j},C$)}
	\label{subsec:ReEnc}
	To generate a first level ciphertext $C^\prime= (k,c_1^\prime,c_2^\prime,c_3^\prime,c_4^\prime,c_5^\prime,c_6^\prime)$ of user $j$ with $rk_{i \xrightarrow[S]{} j}=( r_1, r_2)$ and $C=(k,c_1, c_2, c_3, c_4, c_5, c_6, c_7, c_8, c_9)$ of user $i$,
	the proxy first checks if $C$'s type $\rho \in S$, and then checks \ref{eq1}, \ref{eq2}, and \ref{eq3}. If true, the proxy perfomes the following steps; otherwise, outputs $\perp$.
	\begin{enumerate}
		\item Set $c_1^\prime=c_1$, $c_4^\prime=c_7$, $c_5^\prime=c_8$, $c_6^\prime = c_9$.
		\item Compute 
		$c_2^\prime=e(c_2,r_1)$.
		\item Compute 
		$c_3^\prime=c_5 \oplus H_1(\frac{e(\prod_{\nu\in S}g_{n+1-\nu },c_4)}{e(r_2\cdot \prod_{\nu\in S,\nu\neq \rho}g_{n+1-\nu +\rho},c_3)},c_1)$.
		\item Send 
		$C^\prime=(k,c_1^\prime,c_2^\prime,c_3^\prime,c_4^\prime,c_5^\prime,c_6^\prime)$ to user $j$.
	\end{enumerate}
	Note that, for ciphertext $C$, the verification of \ref{eq1} and \ref{eq2} can be alternately done by picking $d_1,d_2\in_R \mathbb{Z}_p^*$ and testing if
	
	\[
	e(c_1,(g^{a_1})^{d_1}\cdot (u^h \cdot v^k \cdot w)^{d_2})=e(c_2^{d_1}\cdot c_6^{d_2},d)\label{eq5}\tag{Eq. 5}
	\]
	
	\subsubsection{Dec\textsubscript{2} ($sk_i,C$)}
	\label{subsec:Dec}
	To decrypt a second level ciphertext $C=(k,c_1, c_2, c_3, c_4, c_5, c_6, c_7, c_8, c_9)$ of user $i$, one performs the following steps:
	
	\begin{enumerate}
		\item Check the validity of the ciphertext via \ref{eq3} and \ref{eq5}. If the verification does not succeed, output $\perp$.
		\item With $sk_i=(a_{i,1},a_{i,2},a_{i,3})$, compute $K=e(c_2,g)^{1/a_{i,1}}$. If $H_1(K,c_8)= c_9$ holds, output $m=H_1(K,c_1)\oplus c_5\oplus H_1(e(c_3,g_{n+1}),c_1)$; otherwise, output $\perp$.
	\end{enumerate}
	
	\subsubsection{Dec\textsubscript{1} ($sk_j,C^\prime$)}
	\label{subsec:Dec}
	To decrypt a first level ciphertext $C^\prime=(k,c_1^\prime, c_2^\prime, c_3^\prime, c_4^\prime, c_5^\prime,c_6^\prime)$ of user $j$, one performs the following steps:
	\begin{enumerate}
		\item Check the validity of the ciphertext as in \ref{eq4}. If the verification does not succeed, output $\perp$.
		\item With  $sk_j=(a_{j,1},a_{j,2},a_{j,3})$, compute $K=(c_2^\prime )^{1/a_{j,1}}$. Output $m=H_1(K,c_1^\prime)\oplus c_3^\prime$ if the following formula holds:
		\[
		H_1(K,c_5^\prime)=c_6^\prime\label{eq6}\tag{Eq. 6}
		\]
		Otherwise, output $\perp$.
	\end{enumerate}
	
	\subsection{Correctness}
	\label{sec:Correct}
	We demonstrate the correctness of the proposed scheme as follows:
	\begin{enumerate}
		
		\item The correctness of the re-encryption, from user $i$'s ciphertext to user $j$'s ciphertext, is shown as below:\\
		$
		\begin{array}{ll}
		
		c_1^\prime = c_1 = d^r.

		\end{array}		 
		$
		\\
		$
		\begin{array}{ll}
		
		c_2^\prime 
		&= e(c_2,r_1)\\
		& = e(g^{a_{i,1}r},g^{a_{j,1}/a_{i,1}}) \\
		& = e(g,g)^{a_{j,1}r}.\\
		
		\end{array}		 
		$
		\\
		$
		\begin{array}{ll}
		
		c_3^\prime 
		& = c_5\oplus H_1(\frac{e(\prod_{\nu\in S}g_{n+1-\nu},c_4)}{e(r_2\cdot \prod_{\nu\in S,\nu\neq \rho}g_{n+1-\nu+\rho},c_3)},c_1)\\
		& = c_5\oplus H_1(\frac{e(\prod_{\nu\in S}g_{n+1-\nu},(g^{a_{i,2}}\cdot g_\rho)^t)}{e(\prod_{\nu\in S}g_{n+1-\nu }^{a_{i,2}}\cdot \prod_{\nu\in S,\nu\neq \rho}g_{n+1-\nu+\rho},g^t)},c_1)\\			
		
		& = c_5\oplus H_1(\frac{e(\prod_{\nu\in S}g_{n+1-\nu},g_\rho^t)}{e(\prod_{\nu\in S,\nu\neq \rho}g_{n+1-\nu+\rho},g^t)},c_1)\\
		
		& = c_5\oplus H_1(\frac{e(\prod_{\nu\in S}g_{n+1-\nu+\rho},g^t)}{e(\prod_{\nu\in S}g_{n+1-\nu+\rho},g^t)/e(g_{n+1},g^t)},c_1)\\		
		
		& = c_5\oplus H_1(e(g_{n+1},g^t),c_1))\\		
		& = H_1(K,c_1)\oplus m.\\
		
		\end{array}		 
		$
		\\
		$
		\begin{array}{ll}
		
		c_4^\prime
		& = c_7 \\
		& = (u^{h^\prime} \cdot v^k \cdot w )^r \\
		& = (u^{H(c_1,H_1(K,c_1)\oplus m)} \cdot v^k \cdot w )^r.
		
		\end{array}		 
		$\\
		$
		\begin{array}{ll}
		
		c_5^\prime &=c_8=g^\eta .\\
		
		\end{array}		 
		$\\
		$
		\begin{array}{ll}
		
		c_6^\prime &=c_9=H_1(K,g^\eta) .\\

		\end{array}		 
		$\\

		\item The correctness of the decryption on a first level ciphertext is demonstrated as below:
		\[
		\begin{array}{cl}
		&K = (c_2^\prime )^{1/a_{j,1}}=e(g^{a_{j,1}},g)^{r/a_{j,1}}=e(g,g)^r,\\
		&H_1(K,c_1^\prime)\oplus c_3^\prime=H_1(K,c_1^\prime) \oplus H_1(K,c_1^\prime)\oplus m\\
		& = m.
		\end{array}
		\]
		where $sk_i = (a_{j,1},a_{j,2},a_{j,3})$.
		\item The correctness of the decryption on a second level ciphertext is as follows:
		\[
		\begin{array}{l}
		K = e(c_2,g)^{1/a_{i,1}}=e(g^{a_{i,1}r},g)^{1/a_{i,1}}=e(g,g)^r,\\
		H_1(K,c_1)\oplus c_5\oplus H_1(e(c_3,g_{n+1}),c_1)= \\ H_1(K,c_1)\oplus H_1(e(g_1,g_n)^t,c_1) \oplus c_5=\\
		H_1(K,c_1)\oplus H_1(K,c_1)\oplus H_1(e(g_1,g_n)^t,c_1)\oplus\\ H_1(e(g_1,g_n)^t,c_1)\oplus m
		= m.
		\end{array} 
		\]
		where $sk_i = (a_{i,1},a_{i,2},a_{i,3})$.
	\end{enumerate}
	\section{Security Proof}
	\label{ch:proof}
	In this section, we prove that our scheme is IND-2PRE-CCA secure at the second level and IND-1PRE-CCA secure at the first level assuming that $H$ and $H_1$ are two collision resistant hash functions, and the 3-wDBDHI assumption holds. Our proof is adapted from Weng's CCA secure proxy re-encryption scheme \cite{Weng_2010}.
	\subsection{Security of a Second Level Ciphertext}
	\label{subsec:second level ciphertext}
	\begin{theorem}
		\label{th5}
		If there is an adversary $\mathcal{A}$ who is able to break the $(t,q_{pk},q_{sk},q_{rk},q_{re},q_{d},\epsilon)$-IND-2PRE-CCA security, then we can build another algorithm $\mathcal{B}$ that can break the $(t^\prime,\epsilon^\prime)$-3-wDBDHI assumption in ($\mathbb{G},\mathbb{G}_T$) with
		$$
		\begin{array}{l}
		\epsilon^\prime \geq \frac{\epsilon}{2\dot{e}(1+q_{sk}+q_{rk})}-\frac{q_d+q_{re}}{p}-Adv^{HASH}_{H,\mathcal{A}}-Adv^{HASH}_{H_1,\mathcal{A}}, \\ t^\prime\leq t+\mathcal{O}(\tau (q_{pk}+q_{rk}+q_{re}+q_{d}))
		\end{array}
		$$
		in which $\tau$ represents the maximum computation time to calculate an exponentiation, a multi-exponentiation and a paring, and $\dot{e}$ represents the base of the natural logarithm. The scenario of the proof is depicted in Fig \ref{fig:proof}.
		
		\begin{figure}[h]
			\centering
			\includegraphics[width=0.5\textwidth]{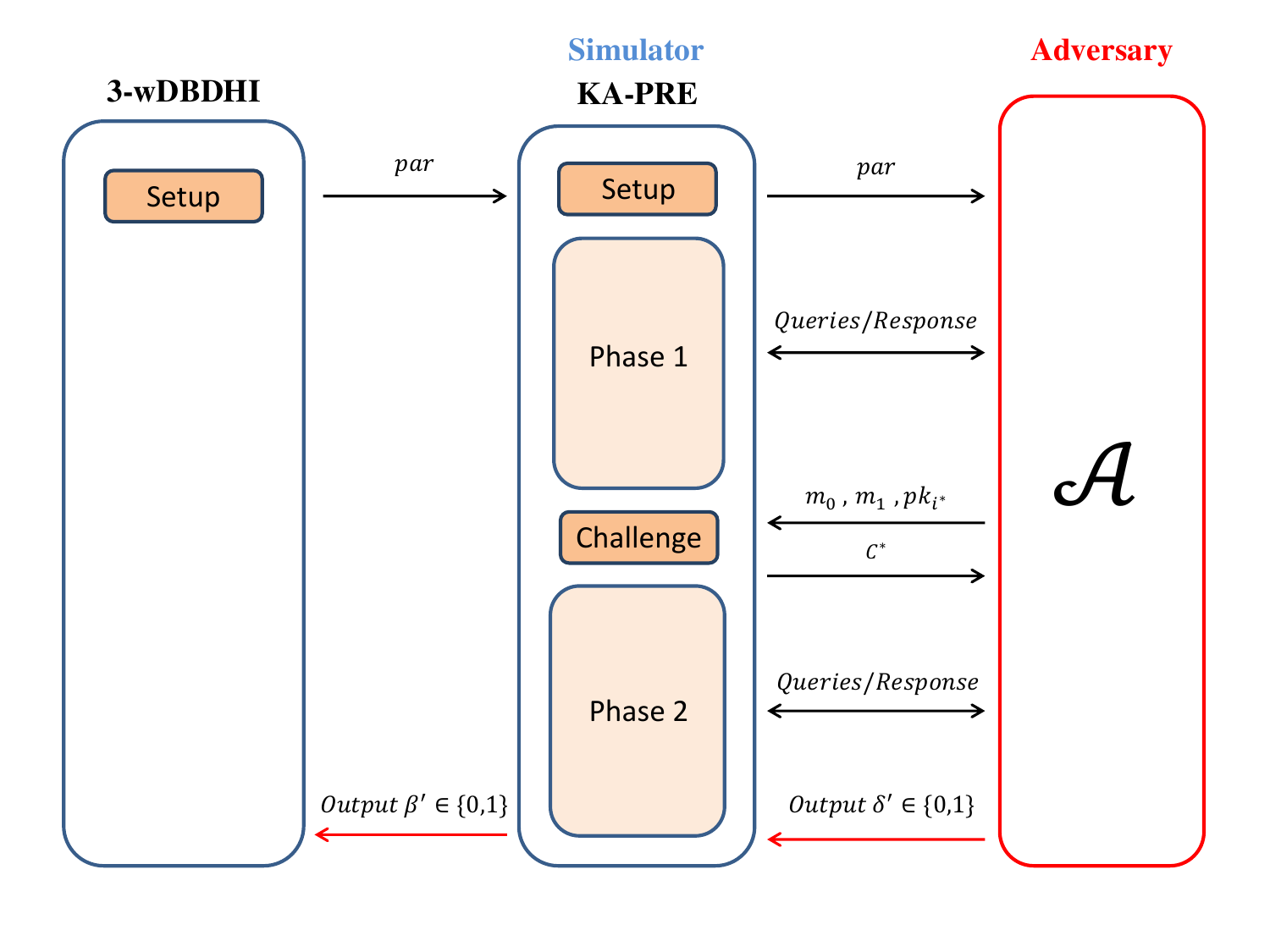}
			\caption{The scenario of the proof}	
			\label{fig:proof}
		\end{figure}
	\end{theorem}
	
	\begin{proof}
		First, we set a 3-wDBDHI instance to be $(g,A_{-1}=g^{1/a},A_1=g^a,A_2=g^{(a^2)},B=g^b,Q)\in \mathbb{G}^5\times \mathbb{G}_T$ with randomly chosen $a,b\in_R \mathbb{Z}_p$. Then we build an algorithm $\mathcal{B}$ who plays an IND-2PRE-CCA game with an adversary $\mathcal{A}$ in order to break the 3-wDBDHI assumption.  
		
		\noindent \textbf{Setup}. $\mathcal{B}$ setups public parameters $d=A_2^{\alpha_0},u=A_1^{\alpha_1}A_2^{\beta_1},v=A_1^{\alpha_2}A_2^{\beta_2},w=A_1^{\alpha_3}A_2^{\beta_3}$ for randomly chosen $\alpha_0,\alpha_1,\alpha_2,\alpha_3,\beta_1,\beta_2,\beta_3 \in_R \mathbb{Z}_p$.
		\\\textbf{Phase 1}. In  phase 1, $\mathcal{A}$ can issue queries in the IND-2PRE-CCA game. $\mathcal{B}$ keeps a table $T$ and answers the queries for $\mathcal{A}$ as follows:
		
		\begin{itemize}
			\item $\textit{Public Key Oracle}\quad \mathcal{O}_{pk}(i):$\\
			$\mathcal{B}$ selects $x_{i,1},x_{i,2},x_{i,3}\in_R \mathbb{Z}_p$. Next, using the Coron's technique \cite{Coron_2000}, we choose a number $s_i\in \{0,1,\lq - \rq\}$ such that $\text{Pr}[s_i=0]=\text{Pr}[s_i=1]=\theta$ and $\text{Pr}[s_i= \lq - \rq]=1-2\theta$ where $\theta$ will be determined later. If $s_i=\lq - \rq$, $\mathcal{B}$ sets $pk_i=(g^{x_{i,1}},g^{x_{i,2}},\Delta_i)$; If $s_i=0$, $\mathcal{B}$ sets $pk_i=(A_2^{x_{i,1}},g^{x_{i,2}},\Delta_i)$; If $s_i=1$, $\mathcal{B}$ sets $pk_i=(A_1^{x_{i,1}},g^{x_{i,2}},\Delta_i)$. Note that $\Delta_i = \langle g_1,g_2,\dots,g_n,g_{n+2},\dots,g_{2n} \rangle$ where $g_\rho = g^{x_{i,3}^{\rho}}$ and $n$ represents the number of the types specified by user $i$.
			$\mathcal{B}$ then adds an entry $(pk_i,x_{i,1},x_{i,2},x_{i,3},s_i)$ into the table $T$ and sends $pk_i$ to $\mathcal{A}$. 
			
			\item $\textit{Secret Key Oracle}\quad \mathcal{O}_{sk}(pk_i):$\\
			On input of user $i$'s public key $pk_i$, $B$ searches the entry $(pk_i,x_{i,1},x_{i,2},x_{i,3},s_i)$ in the table $T$. If $s_i= \lq - \rq$, it returns $sk_i=(x_{i,1},x_{i,2},x_{i,3})$ to $\mathcal{A}$. Otherwise, $\mathcal{B}$ outputs $\perp$ and terminates the process.
			
			\item $\textit{Re-encryption key oracle}\quad \mathcal{O}_{rk}(S,pk_i,pk_j):$\\
			$\mathcal{B}$ first accesses entries  $(pk_i,x_{i,1},x_{i,2},x_{i,3},s_i)$ and $(pk_j,x_{j,1},x_{j,2},x_{j,3},s_j)$ from table $T$ and then generates the re-encryption key $ rk_{i \xrightarrow[S]{} j}$ for $\mathcal{A}$ according to the following cases:
			\begin{itemize}
				\item[-] {$s_i = \lq - \rq $} : It means that $sk_i=(x_{i,1},x_{i,2},x_{i,3})$. $\mathcal{B}$ returns $ rk_{i \xrightarrow[S]{} j}=(pk_{j,1}^{1/x_{i,1}},\prod_{\nu\in S}g_{n+1-\nu}^{x_{i,2}})$.
				
				\item[-] {$s_i = s_j $} : $\mathcal{B}$ returns  $ rk_{i \xrightarrow[S]{} j}=(g^{x_{j,1}/x_{i,1}},\prod_{\nu \in S}g_{n+1-\nu }^{x_{i,2}})$.
				
				\item[-] {$s_i = 1 \wedge s_j = 0 $} : It means that $sk_{i,1}=ax_{i,1}$ and $sk_{j,1}=a^2x_{j,1}$.  $\mathcal{B}$ returns  $ rk_{i \xrightarrow[S]{} j}=(A_1^{x_{j,1}/x_{i,1}},\prod_{\nu\in S}g_{n+1-\nu}^{x_{i,2}})$.
				
				\item[-] {$s_i = 1 \wedge s_j = \lq - \rq $} : It means that $sk_{i,1}=ax_{i,1}$ and $sk_{j,1}=x_{j,1}$.  $\mathcal{B}$ returns  $ rk_{i \xrightarrow[S]{} j}=(A_{-1}^{x_{j,1}/x_{i,1}},\prod_{\nu\in S}g_{n+1-\nu }^{x_{i,2}})$.
				
				\item[-] {$s_i = 0 \wedge s_j = 1 $} : It means that $sk_{i,1}=a^2x_{i,1}$ and $sk_{j,1}=ax_{j,1}$.  $\mathcal{B}$ returns  $ rk_{i \xrightarrow[S]{} j}=(A_{-1}^{x_{j,1}/x_{i,1}},\prod_{\nu\in S}g_{n+1-\nu }^{x_{i,2}})$.
				
				\item[-] {$s_i = 0 \wedge s_j = \lq - \rq $} : $\mathcal{B}$ outputs $\perp$ and terminates the process.
				
			\end{itemize}

			\item $\textit{Re-encryption oracle}\quad \mathcal{O}_{re}(S,pk_i,pk_j,C):$\\
			$\mathcal{B}$ first accesses $C=(k,c_1,c_2,c_3,c_4,c_5,c_6,c_7,c_8,c_9)$ and checks if $C$'s type $\rho\in S$. Then, it checks the validity of the ciphertext via \ref{eq3}, and \ref{eq5}. If the verification does not succeed, it outputs $\perp$ to $\mathcal{A}$. Otherwise, $\mathcal{B}$ does the following steps:
			\begin{enumerate}
				\item Parse tuples  $(pk_i,x_{i,1},x_{i,2},x_{i,3},s_i)$ and $(pk_j,x_{j,1},x_{j,2},x_{j,3},s_j)$ from table $T$.
				\item If $(s_i = 0 \wedge s_j = \lq - \rq)$ holds, it means that $sk_{i,1}=a^2x_{i,1}$ and $sk_{j,1}=x_{j,1}$. From
				$c_1=d^r=A_2^{r \cdot \alpha_0}$ and $c_4=(u^h \cdot v^k \cdot  w)^r=(A_1^{\alpha_1 h + \alpha_2 k + \alpha_3} \cdot A_2^{\beta_1h+\beta_2k+\beta_3})^r$ where $h=H(c_1,c_3)$, $\mathcal{B}$ can compute:
				\[
				A_1^r=\left(\frac{c_6}{c_1^{\frac{\beta_1h+\beta_2k+\beta_3}{\alpha_0}}}\right)^{\frac{1}{\alpha_1 h + \alpha_2 k + \alpha_3}}
				\] 
				Then $\mathcal{B}$ can compute $K=e(A_{-1},A_1^r)=e(g,g)^r$, set $c_1^\prime=c_1,c_2^\prime=K^{x_{j,1}}=e(pk_j,g)^r,c_3^\prime=c_3 \oplus H_1(e(g_1,g_n)^t,c_1),c_4^\prime=c_7, c_5^\prime = c_8, c_6^\prime = c_9$ and return $C^\prime=(k,c_1^\prime,c_2^\prime,c_3^\prime,c_4^\prime, c_5^\prime,c_6^\prime)$ to $\mathcal{A}$.
				\item Otherwise, output the re-encryption key $rk_{i \xrightarrow[S]{} j}$ by querying the re-encryption key oracle $\mathcal{O}_{rk}$, and then return ReEnc($rk_{i \xrightarrow[S]{} j}, C $) to $\mathcal{A}$.
				
			\end{enumerate}
			
			Note that the public parameters $u=A_1^{\alpha_1}A_2^{\beta_1},v=A_1^{\alpha_2}A_2^{\beta_2},w=A_1^{\alpha_2}A_2^{\beta_2}$, $\alpha_1,\alpha_2$ and $\alpha_3$ are blinded by $\beta_1,\beta_2$ and $\beta_3$. Therefore, no information about  $\alpha_1,\alpha_2$ and $\alpha_3$ is presented to the adversary. The equality $\alpha_1 h + \alpha_2 k + \alpha_3 =0$ mod $p$ information-theoretically holds with probability $\frac{1}{p}$.
			
			\item $\textit{First level decryption oracle}\quad \mathcal{O}_{1d}(pk_j,C^\prime):$\\
			$\mathcal{B}$ first accesses $C^\prime=(k,c_1^\prime,c_2^\prime,c_3^\prime,c_4^\prime,c_5^\prime,c_6^\prime)$. Next, it recovers $(pk_j,x_{j,1},x_{j,2},x_{j,3},s_j)$ from table $T$. If $s_j=\lq - \rq$, it means that $sk_{j,1}=x_{j,1}$. $\mathcal{B}$ sends Dec\textsubscript{1}($sk_j,C^\prime$) to $\mathcal{A}$. Otherwise, $\mathcal{B}$ does the following steps:
			\begin{enumerate}
				\item Compute $h^\prime =H(c_1^\prime,c_3^\prime)$ and validate the ciphertext via \ref{eq4}. If the verification does not succeed, output $\perp$.
				\item Compute $A_1^r=\left(\frac{c_4^\prime}{(c_1^\prime)^{\frac{\beta_1h^\prime+\beta_2k+\beta_3}{\alpha_0}}}\right)^{\frac{1}{\alpha_1 h^\prime + \alpha_2 k + \alpha_3}}$ and $K=e(A_{-1},A_1^r)=e(g,g)^r$. Note that, similar to the re-encryption oracle $\mathcal{O}_{re}$, probability of $\alpha_1h^\prime+\alpha_2k+\alpha_3=0$ mod $p$ is $\frac{1}{p}$.
				\item If $H_1(K,c_5^\prime)\neq c_6^\prime$, output $\perp$. Otherwise, output $m=H_1(K,c_1^\prime)\oplus c_3^\prime$.
			\end{enumerate}
			
			If $\alpha_1 h^\prime + \alpha_2 k + \alpha_3=0$ mod $p$, $\mathcal{B}$ can reject this invalid first level ciphertexts. For $C^\prime=(k,c_1^\prime,c_2^\prime,c_3^\prime,c_4^\prime,c_5^\prime,c_6^\prime)$ under public key $pk_j$, the validity of $k,c_1^\prime,c_3^\prime$ can be ensured by \ref{eq4}. Thus, $\mathcal{B}$ only needs to validate $c_2^\prime$. Suppose $c_1^\prime=d^r,c_3^\prime=H_1(K,c_1^\prime) \oplus m$ and $c_4^\prime=(u^{H(c_1^\prime,c_3^\prime)}\cdot v^k \cdot w)^r$, where $K=e(g,g)^r$. To validate $c_2^\prime$, $\mathcal{B}$ has to check whether $c_2^\prime=e(pk_{j,1},g)^r$ holds. Fortunately, $\mathcal{B}$ can compute $e(pk_{j,1},g)^r$ according to the following cases:
			\begin{itemize}
				\item [-] {$s_j = 1 $} (it means that $pk_{j,1}=A_1^{x_{j,1}}$): $\mathcal{B}$ can obtain $e(pk_{j,1},g)^r$  by computing
				$$
				\begin{array}{l}
				e(c_1^\prime,A_{-1})^{\frac{x_{j,1}}{\alpha_0}}=e(d^r,A_{-1})^{\frac{x_{j,1}}{\alpha_0}}= \\e(A_2^{\alpha_0r},A_{-1})^{\frac{x_{j,1}}{\alpha_0}} = e(A_1^{x_{j,1}},g)=e(pk_{j,1},g)^r.
				\end{array}
				$$
				\item [-] {$s_j = 0 $} (it means that $pk_{j,1}=A_2^{x_{j,1}}$): $\mathcal{B}$ can obtain $e(pk_{j,1},g)^r$  by computing
				$$
				\begin{array}{l}
				e(c_1^\prime,g)^{\frac{x_{j,1}}{\alpha_0}}=e(d^r,g)^{\frac{x_{j,1}}{\alpha_0}}\\ =e(A_2^{\alpha_0r},g)^{\frac{x_{j,1}}{\alpha_0}} =e(A_1^{x_{j,1}},g)=e(pk_{j,1},g)^r.			
				\end{array}
				$$
			\end{itemize}
			\noindent \textbf{Challenge}. When \textbf{Phase 1} stage is over, $\mathcal{A}$  sends a public key $pk_{i^*}$ and two messages $m_0,m_1 \in \mathbb{G}$ to $\mathcal{B}$ with the restrictions specified in the IND-2PRE-CCA game. Note that both messages $m_0,m_1$ are of type $\rho$.  $\mathcal{B}$ responds as follows:
			\begin{enumerate}
				\item Access $(pk_{i^*},x_{i^*,1},x_{i^*,2},x_{i^*,3},s_{i^*})$ from table $T$. If $s_{i^*}\neq 0$, $\mathcal{B}$ outputs $\perp$ and terminates the process. Otherwise, it means that $pk_{i^*,1}=A_2^{x_{i^*,1}}$, and $\mathcal{B}$ proceeds to execute the rest steps.
				\item Pick $\delta \in \{0,1\}$ and $t,\eta\in_R \mathbb{Z}_p$. Define $k^*=-\frac{\alpha_1h^*+\alpha_3}{\alpha_2},c_1^*=B^{\alpha_0},c_2^*=B^{x_{i^*,1}},c_3^*=g^t,c_4^*=(g^{x_{i^*,2}}\cdot g_\rho)^t,c_5^*=H_1(Q,c_1^*)\oplus m_\delta \oplus H_1(e(g_1,g_n)^t,c_1^*),c_6^*=B^{(\beta_1h^*+\beta_2k^*+\beta_3)}$, $c_7^*=B^{(\beta_1h^{*\prime}+\beta_2k^*+\beta_3)}$, $c_8^*= g^\eta$,$c_9^*=H_1(Z^{r^*},c_8^*)$ where $h^*=H(c_1^*,c_5^*)$ and $h^{*\prime}=H(c_1^*,H_1(Q,c_1^*)\oplus m_\delta)$. Return $C^*=(k^*,c_1^*,c_2^*,c_3^*,c_4^*,c_5^*,c_6^*,c_7^*,c_8^*,c_9^*)$ as the
				challenge ciphertext to $\mathcal{A}$.
				
			\end{enumerate}
			Observe that, if $Q=e(g,g)^{\frac{b}{a^2}}$, $C^*$ is indeed a valid challenge ciphertext under public key $pk_{i^*}$. To see this, setting $r^*=\frac{b}{a^2}$, we have
			
			$$
			\begin{array}{cl}
			c_1^* &= B^{\alpha_0} = (g^{a^2})^{\alpha_0 \cdot \frac{b}{a^2}} = (A_2^{\alpha_0})^{r^*} = d^{r^*}\\
			c_2^* &= B^{x_{i^*}} = (g^{a^2})^{x_{i^*} \cdot \frac{b}{a^2}} =  (A_2^{x_{i^*}})^{r^*} = pk_{i^*,1}^{r^*}\\
			c_3^* &= g^t\\
			c_4^* &=(g^{x_{i^*,2}}\cdot g_\rho)^t\\
			c_5^* &= H_1(Q,c_1^*)\oplus m_\delta \oplus H_1(e(g_1,g_n)^t,c_1^*)\\
			&=H_1(Z^{r^*},c_1^*)\oplus m_\delta \oplus H_1(e(g_1,g_n)^t,c_1^*)\\
			c_6^* &= B^{\beta_1h^*+\beta_2k^*+\beta_3} = \left(A_2^{\beta_1h^*+\beta_2k^*+\beta_3}\right)^{r^*} \\
			&= \left(A_1^{\alpha_1h^*}A_1^{-\frac{\alpha_1h^*+\alpha_3}{\alpha_2}\cdot\alpha_2}A_1^{\alpha_3}A_2^{\beta_1h^*+\beta_2k^*+\beta_3}\right)^{r^*}\\
			&=\left(A_1^{\alpha_1h^*}A_1^{k^*\cdot \alpha_2}A_1^{\alpha_3}A_2^{\beta_1h^*+\beta_2k^*+\beta_3}\right)^{r^*}\\
			&= \left((A_1^{\alpha_1}A_2^{\beta_1})^{h^*}\cdot (A_1^{\alpha_2}A_2^{\beta_2})^{k^*}\cdot (A_1^{\alpha_1}A_2^{\beta_1})\right)^{r^*}\\
			&=\left(u^{h^*}\cdot v^{k^*} \cdot w\right)^{r^*}\\
			c_7^* &= B^{\beta_1h^{\prime*}+\beta_2k^*+\beta_3} = \left(A_2^{\beta_1h^{\prime*}+\beta_2k^*+\beta_3}\right)^{r^*}\\
			&= \left(A_1^{\alpha_1h^{\prime*}}A_1^{-\frac{\alpha_1h^{\prime*}+\alpha_3}{\alpha_2}\cdot\alpha_2}A_1^{\alpha_3}A_2^{\beta_1h^{\prime*}+\beta_2k^*+\beta_3}\right)^{r^*}\\
			&=\left(A_1^{\alpha_1h^{\prime*}}A_1^{k^*\cdot \alpha_2}A_1^{\alpha_3}A_2^{\beta_1h^{\prime*}+\beta_2k^*+\beta_3}\right)^{r^*}\\
			&= \left((A_1^{\alpha_1}A_2^{\beta_1})^{h^{\prime*}}\cdot (A_1^{\alpha_2}A_2^{\beta_2})^{k^*}\cdot (A_1^{\alpha_1}A_2^{\beta_1})\right)^{r^*}\\
			&=\left(u^{h^{\prime*}}\cdot v^{k^*} \cdot w\right)^{r^*}\\
			c_8^* &= g^\eta\\
			c_9^* &= H_1(Q,c_8^*) = H_1(Z^{r^*},c_8^*)\\
			\end{array}$$
			\textbf{Phase 2}. $\mathcal{B}$ responds the queries for $\mathcal{A}$ as that in the \textbf{Phase 1}. $\mathcal{A}$ follows the restrictions described in the IND-2PRE-CCA game. Note that although the $C^*$ exposes the information $\alpha_1h^*+\alpha_2k^*+\alpha_3=0$ to $\mathcal{A}$, the probability of event causing $\mathcal{B}$ to abort is $\frac{1}{p}$.\\
			
			\noindent \textbf{Output}. Finally, $\mathcal{A}$ returns $\delta^\prime\in \{0,1\}$. If $\delta^\prime=\delta$, $\mathcal{B}$ returns $\beta^\prime=1$; else, returns $\beta^\prime=0$. We can see that the simulation of oracle $\mathcal{O}_{pk}$ is perfectly simulated. Here, \textbf{terminate} stands for $\mathcal{B}$'s aborting in oracles $\mathcal{O}_{sk},\mathcal{O}_{rk}$ or in the \textbf{Challenge}. We have $\Pr[$\textbf{terminate}$]=(1-2\theta)^{q_{sk}}(\theta(1-2\theta))^{q_{rk}}\theta\geqslant (1-2\theta)^{q_{sk}+q_{rk}}\theta$, which is maximized at $\theta_{opt}=\frac{q_{sk}+q_{rk}}{2(1+q_{sk}+q_{rk})}$. With $\theta_{opt}$, the probability $\Pr[$\textbf{terminate}$]$ is at least $\frac{1}{2\dot{e}(1+q_{sk}+q_{rk})}$. The simulation for oracles $\mathcal{O}_{sk},\mathcal{O}_{rk}$ and the challenge ciphertext are perfectly simulated if \textbf{terminate} does not happen. The simulation of $\mathcal{O}_{re}$ is also perfect, unless $\alpha_1h +\alpha_2k+\alpha_3=0$ mod $p$ happens (denoting this event by \textbf{ReEErr}). Nonetheless, as aforementioned, the formula $\alpha_1h+\alpha_2k+\alpha_3=0$ mod $p$ holds in every query with probability $\frac{1}{p}$. Therefore, we have $\Pr[$\textbf{ReEErr}$]\leqslant\frac{q_{re}}{p}$.\\
			The decryption oracle $\mathcal{O}_{1d}$ is also perfectly simulated, unless $\alpha_1h+\alpha_2k+\alpha_3=0$ mod $p$ happens during the simulation (denoting this event by \textbf{DecErr}). We have $\Pr[$\textbf{DecErr}$]\leqslant\frac{q_d}{p}$. Then we know that $\mathcal{B}$'s advantage against the 3-wDBDHI assumption
			is $\epsilon^\prime\geqslant\frac{\epsilon}{2\dot{e}(1+q_{sk}+q_{rk})}-\frac{q_d+q_{re}}{p}-Adv^{HASH}_{H,\mathcal{A}}-Adv^{HASH}_{H_1,\mathcal{A}}$, and $\mathcal{B}$'s running time is $t^\prime\leqslant t+\mathcal{O}(\tau (q_{pk}+q_{rk}+q_{re}+q_{d})).$  The results complete the proof of a second level ciphertext.
			
		\end{itemize}
	\end{proof}
	
	\subsection{Security of a First Level Ciphertext}
	\label{subsec:first ciphertext}
	\begin{theorem}
		\label{th5}
		If there is an adversary $\mathcal{A}$ who is able to break the $(t,q_{pk},q_{sk},q_{rk},q_{d},\epsilon)$-IND-1PRE-CCA security, then we can build an algorithm $\mathcal{B}$ that can break the $(t^\prime,\epsilon^\prime)$-3-wDBDHI assumption in ($\mathbb{G},\mathbb{G}_T$) with
		$$
		\begin{array}{l}
		\epsilon^\prime \geq \frac{\epsilon}{\dot{e}(1+q_{sk})}-\frac{q_d}{p}-Adv^{HASH}_{H,\mathcal{A}}-Adv^{HASH}_{H_1,\mathcal{A}},\\ t^\prime\leq t+\mathcal{O}(\tau (q_{pk}+q_{rk}+q_{d}))
		\end{array}
		$$
		in which $\tau$ represents the maximum computation time to calculate an exponentiation, a multi-exponentiation and a paring, and $\dot{e}$ represents the base of the natural logarithm. The scenario of the proof is depicted in Fig \ref{fig:proof}.
	\end{theorem}
	
	\begin{proof}
		First, we set a 3-wDBDHI instance to be $(g,A_{-1}=g^{1/a},A_1=g^a,A_2=g^{(a^2)},B=g^b,Q)\in \mathbb{G}^5\times \mathbb{G}_T$ with randomly chosen $a,b\in_R \mathbb{Z}_p$. Then we build an algorithm $\mathcal{B}$ who plays a IND-1PRE-CCA game with an adversary $\mathcal{A}$ in order to break the 3-wDBDHI assumption.
		
		\noindent \textbf{Setup}.  $\mathcal{B}$ setups public parameters $d=A_2^{\alpha_0},u=A_1^{\alpha_1}A_2^{\beta_1},v=A_1^{\alpha_2}A_2^{\beta_2},w=A_1^{\alpha_3}A_2^{\beta_3}$ for randomly chosen $\alpha_0,\alpha_1,\alpha_2,\alpha_3,\beta_1,\beta_2,\beta_3 \in_R \mathbb{Z}_p$.
		\\\textbf{Phase 1}. In  phase 1, $\mathcal{A}$ can issues queries in the IND-1PRE-CCA game. $\mathcal{B}$ keeps a table $T$ and answers the queries for $\mathcal{A}$ as follows:
		
		\begin{itemize}
			\item $\textit{Public Key Oracle}\quad \mathcal{O}_{pk}(i):$\\
			$\mathcal{B}$ selects $x_{i,1},x_{i,2},x_{i,3}\in_R \mathbb{Z}_p$. Next, we choose a number $s_i\in \{0,1\}$. If $s_i=0$, $\mathcal{B}$ sets $pk_i=(A_1^{x_{i,1}},g^{x_{i,2}},\Delta_i)$; If $s_i=1$, $\mathcal{B}$ sets $pk_i=(g^{x_{i,1}},g^{x_{i,2}},\Delta_i)$. Note that $\Delta_i = \langle g_1,g_2,\dots,g_n,g_{n+2},\dots,g_{2n} \rangle$ where $g_\rho = g^{x_{i,3}^{\rho}}$ and $n$ represents the number of the types specified by user $i$.
			$\mathcal{B}$ then adds an entry $(pk_i,x_{i,1},x_{i,2},x_{i,3},s_i)$ into the table $T$ and sends $pk_i$ to $\mathcal{A}$. 		
			\item $\textit{Secret Key Oracle}\quad \mathcal{O}_{sk}(pk_i):$\\
			On input of user $i$'s public key $pk_i$, $B$ searches the entry $(pk_i,x_{i,1},x_{i,2},x_{i,3},s_i)$ in the table $T$. If $s_i= 1$, it returns $sk_i=(x_{i,1},x_{i,2},x_{i,3})$ to $\mathcal{A}$. Otherwise, $\mathcal{B}$ outputs $\perp$ and terminates the process.
			
			\item $\textit{Re-encryption key oracle}\quad \mathcal{O}_{rk}(S,pk_i,pk_j):$\\
			$\mathcal{B}$ first accesses entries  $(pk_i,x_{i,1},x_{i,2},x_{i,3},s_i)$ and $(pk_j,x_{j,1},x_{j,2},x_{j,3},s_j)$ from table $T$ and then generates the re-encryption key $ rk_{i \xrightarrow[S]{} j}$ for $\mathcal{A}$ according to the following cases:
			\begin{itemize}
				\item[-] {$s_i = 1 $} : It means that $sk_i=(x_{i,1},x_{i,2},x_{i,3})$. $\mathcal{B}$ returns $ rk_{i \xrightarrow[S]{} j}=(pk_{j,1}^{1/x_{i,1}},\prod_{\nu\in S}g_{n+1-\nu }^{x_{i,2}})$.
				
				\item[-] {$s_i = s_j = 0 $} : $\mathcal{B}$ returns  $ rk_{i \xrightarrow[S]{} j}=(g^{x_{j,1}/x_{i,1}},\prod_{\nu\in S}g_{n+1-\nu }^{x_{i,2}})$.
				
				\item[-] {$s_i = 0 \wedge s_j = 1 $} : It means that $sk_{i,1}=a^2x_{i,1}$ and $sk_{j,1}=ax_{j,1}$.  $\mathcal{B}$ returns  $ rk_{i \xrightarrow[S]{} j}=(A_{-1}^{x_{j,1}/x_{i,1}},\prod_{\nu\in S}g_{n+1-\nu }^{x_{i,2}})$.
			\end{itemize}
			
			\item $\textit{First level decryption oracle}\quad \mathcal{O}_{1d}(pk_j,C^\prime):$\\
			$\mathcal{B}$ first accesses $C^\prime=(k,c_1^\prime,c_2^\prime,c_3^\prime,c_4^\prime,c_5^\prime,c_6^\prime)$. Next, it recovers $(pk_i,x_{i,1},x_{i,2},x_{i,3},s_i)$ from table $T$. If $s_j=1$, it means that $sk_{j,1}=x_{j,1}$ and $\mathcal{B}$ returns Dec\textsubscript{1}($sk_j,C^\prime$) to $\mathcal{A}$. Otherwise, $\mathcal{B}$ does the following steps:
			\begin{enumerate}
				\item Compute $h^\prime=H(c_1^\prime,c_3^\prime)$ and check the validity of the ciphertext by \ref{eq4}. If the verification fails, output $\perp$ indicating an invalid ciphertext; else continue to execute the rest of the steps.
				\item Compute $A_1^r=\left(\frac{c_6}{(c_1^\prime)^{\frac{\beta_1h^\prime+\beta_2k+\beta_3}{\alpha_0}}}\right)^{\frac{1}{\alpha_1 h^\prime + \alpha_2 k + \alpha_3}}$ and $K=e(A_{-1},A_1^r)=e(g,g)^r$. Note that, similar to the analysis in the re-encryption oracle $\mathcal{O}_{re}$, $\alpha_1h^\prime+\alpha_2k+\alpha_3=0$ mod $p$ is with probability at most $\frac{1}{p}$.
				\item If $H_1(K,c_5^\prime)\neq c_6^\prime$, output $\perp$  indicating an invalid ciphertext. Otherwise, output $m=H_1(K,c_1^\prime)\oplus c_3^\prime$.

			\end{enumerate}
			If $\alpha_1 h^\prime + \alpha_2 k + \alpha_3=0$ mod $p$, $\mathcal{B}$ can reject this invalid first level ciphertexts. For $C^\prime=(k,c_1^\prime,c_2^\prime,c_3^\prime,c_4^\prime,c_5^\prime,c_6^\prime)$ under public key $pk_j$, the validity of $k,c_1^\prime,c_3^\prime$ can be ensured by \ref{eq4}. So, $\mathcal{B}$ only needs to check the validity of $c_2^\prime$. Suppose $c_1^\prime=d^r,c_3^\prime=H_1(K,c_1^\prime)\oplus m$ and $c_4^\prime=(u^{H(c_1^\prime,c_3^\prime)}\cdot v^k \cdot w)^r$, where $K=e(g,g)^r$. To
			validate $c_2^\prime$, $\mathcal{B}$ has to check whether $c_2^\prime=e(pk_{j,1},g)^r$ holds. Fortunately, $\mathcal{B}$ can compute $e(pk_{j,1},g)^r$ according to the following:
			$$
			\begin{array}{l}
			e(c_1^\prime,A_{-1})^{\frac{x_{j,1}}{\alpha_0}}= e(d^r,A_{-1})^{\frac{x_{j,1}}{\alpha_0}}=\\ e(A_2^{\alpha_0r},A_{-1})^{\frac{x_{j,1}}{\alpha_0}}=e(A_1^{x_{j,1}},g)=e(pk_{j,1},g)^r.
			\end{array}			
			$$
			
			\noindent \textbf{Challenge}. $\mathcal{A}$ outputs a public key $pk_{i^*}$ and two messages $m_0,m_1\in \mathbb{G}$ with the restrictions specified in the IND-1PRE-CCA game. $\mathcal{B}$ first accesses $(pk_{i^*},x_{i^*,1},x_{i^*,2},x_{i^*,3},s_{i^*})$ from table $T$.  If $s_{i^*}=1$, $\mathcal{B}$ returns $\perp$ and terminates the process. Otherwise, $\mathcal{B}$ picks $\delta\in_R\{0,1\}, t,\eta\in_R \mathbb{Z}_p$, and defines  $c_1^{\prime*}=B^{\alpha_0},c_2^{\prime*}=B^{x_{i,1}^*},c_3^{\prime*}=H_1(Q,c_1^{\prime*})\oplus m_\delta,k^*=-\frac{\alpha_1h^{\prime*}+\alpha_3}{\alpha_2}$ $,c_4^{\prime*}=B^{\beta_1h^{\prime*}+\beta_2k^*+\beta_3},c_5^{\prime*} = g^\eta,c_6^{\prime*} = H_1(Q,c_5^{\prime*})$ where $h^{\prime*}=H(c_1^{\prime*},c_3^{\prime*})$. Then return $C^{\prime*}=(k^*,c_1^{\prime*},c_2^{\prime*},c_3^{\prime*},c_4^{\prime*},c_5^{\prime*},c_6^{\prime*})$   to $\mathcal{A}$.\\
			Note that, if $Q=e(g,g)^{\frac{b}{a^2}}$, $C^{\prime*}$ is actually a valid challenge ciphertext. To verify this, letting $r^*=\frac{b}{a^2}$, the well-formedness of $c_1^{\prime*},c_3^{\prime*}$ and $c_4^{\prime*}$ can be seen as that in the proof second level security; while for $c_2^*$, its well-formedness is as shown below:
			$$
			\begin{array}{l}
			c_2^{\prime*}=e(A_{-1},B)^{x_{i,1}^*}=  e(g^{1/a},g^b)^{x_{i,1}^*}= \\ e(g^{a \cdot x_{i,1}^*} ,g)^{\frac{b}{a^2}}=e(pk_{i,1}^*,g)^{r^*}
			\end{array}			
			$$
			
			\noindent \textbf{Phase 2}. $\mathcal{B}$ responds the queries for $\mathcal{A}$ as that in the \textbf{Phase 1}. $\mathcal{A}$ follows the restrictions described in the IND-1PRE-CCA game. \\
			
			\noindent \textbf{Output}. Finally $\mathcal{A}$ outputs $\delta^\prime\in \{0,1\}$. If $\delta^\prime=\delta$, $\mathcal{B}$ returns $\beta^\prime=1$; else, returns $\beta^\prime=0$. We can see that the simulations of oracle $\mathcal{O}_{pk}$ and $\mathcal{O}_{rk}$ are perfectly simulated. The decryption $\mathcal{O}_{1d}$ is perfectly simulated, unless $\alpha_1h^\prime +\alpha_2k+\alpha_3=0$ mod $p$ happens (denoting this event by \textbf{DecErr}). Similar to the proof in the second level security, we have $\Pr[$\textbf{DecErr}$]\leqslant\frac{q_d}{p}$. Here, \textbf{terminate} stands for $\mathcal{B}$'s aborting in oracles $\mathcal{O}_{sk}$ or in the Challenge phase. We have $\Pr[$\textbf{terminate}$]\leqslant\frac{1}{\dot{e}(1+q_{sk})}$. Then we know that $\mathcal{B}$'s advantage against the 3-wDBDHI assumption is $\epsilon^\prime\geqslant\frac{\epsilon}{\dot{e}(1+q_{sk})}-\frac{q_d}{p}-Adv^{HASH}_{H,\mathcal{A}}-Adv^{HASH}_{H_1,\mathcal{A}}$, and $\mathcal{B}$'s running time is $t^\prime\leqslant t+\mathcal{O}(\tau (q_{pk}+q_{rk}+q_{d}))$.  The results completes the proof of a first level ciphertext.
			
		\end{itemize}
	\end{proof}
	
	\subsection{Master Secret Security}
	\begin{theorem} Our key-aggregate proxy re-encryption scheme is MSS-PRE secure if the 3-wDBDHI assumption
		holds in groups ($\mathbb{G},\mathbb{G}_T$). 
	\end{theorem}
	\begin{proof}
		\noindent Assuming that an adversary $\mathcal{A}$ can recover the secret key $sk_i=(x_{i,1},x_{i,2},x_{i,3})$ of a targeted user $i$ with non-negligible probability $\epsilon$, we can simply simulate an adversary $\mathcal{B}$ that breaks $\Gamma$'s IND-1PRE-CCA security or IND-2PRE-CCA security. That is, for a single-use unidirectional proxy re-encryption scheme, the master secret security is implied by the first level ciphertext security and second level ciphertext security.
	\end{proof}
	
	\section{Implementation and Simulation}
	\label{ch:Implementation}
	In this section, we demonstrate the testing environment and the simulation result of our proposed scheme. A brief introduction of the PBC library is also included. Finally, we show the cost time of the algorithms in the proposed KAPRE scheme. 
	The source code of our implementation is available at https://github.com/ferranschen/KAPRE.
	\subsection{Enviroment}
	\label{sec:contribution}
	The testing environment is shown in the following table.
	
	\begin{table}[h!]
		\begin{center}
			\caption{Testing Environment}
			\label{tab:table1}
			\begin{tabular}{l|l}
				
				OS & 64-bit Ubuntu 16.04.4 LTS \\
				CPU & AMD Ryzen 5 1600 Six-Core Processor \\
				RAM & 3.9 GB \\
				DISK & 32.6 GB \\ 
			\end{tabular}
		\end{center}
	\end{table}
	
	\subsection{The Pairing-Based Cryptography Library}
	\label{sec:pbc}
	The PBC (Pairing-Based Cryptography) library is a free library written in C. Users can easily use this library without the need to know too much about elliptic curves or number theory. 
	Researchers or software developers can effortlessly use this library to architect a cryptographic system without having to build their own wheels. 
	In addition, this library is based on GMP (GNU
	Multiple Precision) library, pairing operation speed is reasonable. 
	For more details, we refer the interested users to https://crypto.stanford.edu/pbc/.
	
	\subsection{Simulation Result}
	\label{sec:result}
	We show the actual execution results of our proposed scheme in Fig \ref{fig:scheme}.
	
	\begin{figure}[h]
		\centering
		\includegraphics[width=0.5\textwidth]{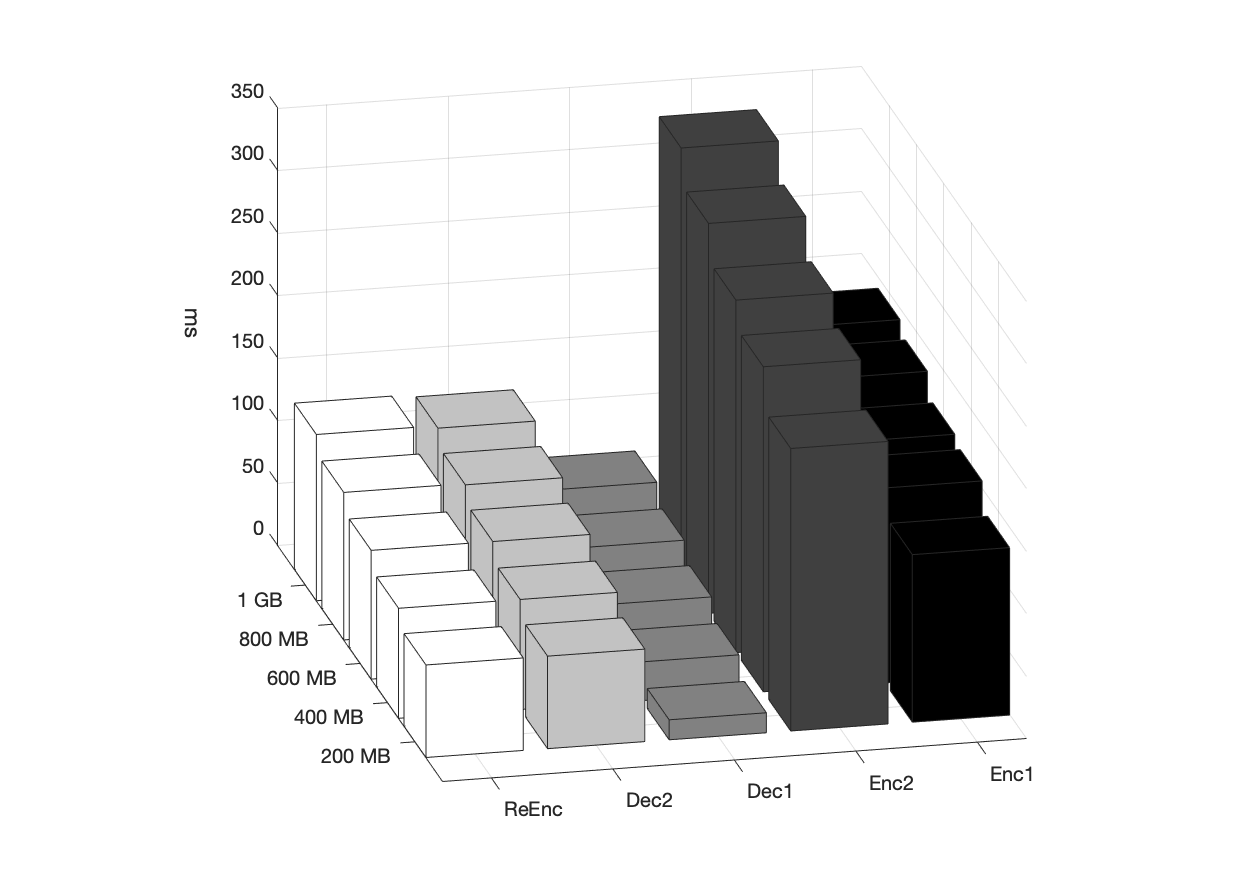}	
		\caption{Computation time of algorithms in our proposed scheme.}
		\label{fig:scheme}
	\end{figure}

	\section{Comparison}
	\label{ch:Comparison}
	
In this section, we make a comprehensive comparison of our proposed scheme with the four related C-PRE schemes \cite{Chu_2009, Fang_2009, Liang_2012, Weng_C_PRE_2009}.
    In TABLE \ref{tab:Properties}, Fig \ref{fig:rekeygencompare}, Fig \ref{fig:rekeylengthcompare}, we compared the speed at which the re-encryption key was made and the storage size it needed. In addition, we also compare the security level, corruption model, standard model, and final-grained access control.
    Fortunately, our scheme is the first CCA secure KAPRE. It is worth mentioning that the proposed scheme achieves fined grained access control using the constant-size re-encryption key.
    As mentioned in \cite{Menezes_2001,Scott_2007,Koblitz_2000,Zhang_2006,Schneier_1999}, we know that $T_p \approx 5T_e, T_s \approx 29 T_m, T_e \approx 240T_m, T_a \approx 0.12T_m$ and $T_h \approx 7.75T_m$.

	\begin{table*}[h]
		\centering
		\begin{footnotesize}	
			\caption{The Comparisons between C-PREs and our scheme.}
			\begin{tabular}{|c|ccccc|}
				\hline
				&Weng \cite{Weng_C_PRE_2009}&Chu \cite{Chu_2009}&Fang \cite{Fang_2009}&Liang \cite{Liang_2012}&Ours\\
				\hline	
				Security &CCA&CPA,RCCA&CCA&CCA&CCA	\\
				\hline	
				Corruption model &Adaptive&Static&Adaptive&Adaptive&Adaptive\\
				\hline	
				Standard model &Yes&Yes&Yes&Yes&Yes	\\
				\hline
				Fine-grained & \multirow{2}{*}{Yes} & \multirow{2}{*}{Yes} & \multirow{2}{*}{Yes} & \multirow{2}{*}{Yes} & \multirow{2}{*}{Yes}\\
				access control &&&&& \\	
				\hline
				ReKeyGen computational
				&\multirow{1}{*}{$2nT_s + nT_h$}
				&\multirow{1}{*}{$2nT_s + 2nT_a$}
				&\multirow{1}{*}{$7nT_s + 3nT_a$}
				&\multirow{1}{*}{$\bigtriangleup  $}
				&\multirow{1}{*}{$2T_s + ( n-1)T_a $}\\
				cost 
				&$\approx 4340n CCs$
				&$\approx 3844n CCs$
				&$\approx 13422n CCs$
				&
				&$ \approx 8n + 3820 CCs$\\
				\hline
				Re-encryption keys 
				&\multirow{2}{*}{$n(\left|\mathbb{G}\right| + \left|\mathbb{G}\right|) = 512n bits$}
				& \multirow{2}{*}{$n(\left|\mathbb{G}\right|) = 256n bits$}
				&\multirow{2}{*}{$n(\left|\mathbb{G}\right| + 3\left|\mathbb{G}\right| ) = 1024nbits$}
				&\multirow{2}{*}{$\bigtriangleup   $}
				& \multirow{2}{*}{$\left|\mathbb{G}\right| + \left|\mathbb{G}\right| = 512 bits$}\\
				Length &&&&&\\
				\hline
				
			\end{tabular}
			
			\label{tab:Properties}
		\end{footnotesize}
		\begin{itemize}
			\item $\left| \mathbb{G}\right|$: the length of an element in $\mathbb{G}$
			\item $\left| \mathbb{G}_T\right|$: the length of an element in $\mathbb{G}_T$
			\item $n$ : the number of total file types
			\item $T_p$ : the cost of a pairing operation
			\item $T_m$ : the cost of a modular multiplication in $Z_q$
			\item $T_s$ : the cost of a scalar multiplication in an additive group or an exponentiation in a multiplicative group
			\item $T_a$ : the cost of an addition in an additive group or a multiplication in a multiplicative group
			\item $T_h$ : the cost of a hash operation
			\item $CCs$ : clock cycles
			\item $\bigtriangleup$ : the performance depends on the cryptographic primitives 
		\end{itemize}
	\end{table*}
	\begin{figure}[h]
		\centering
		\includegraphics[width=0.5\textwidth]{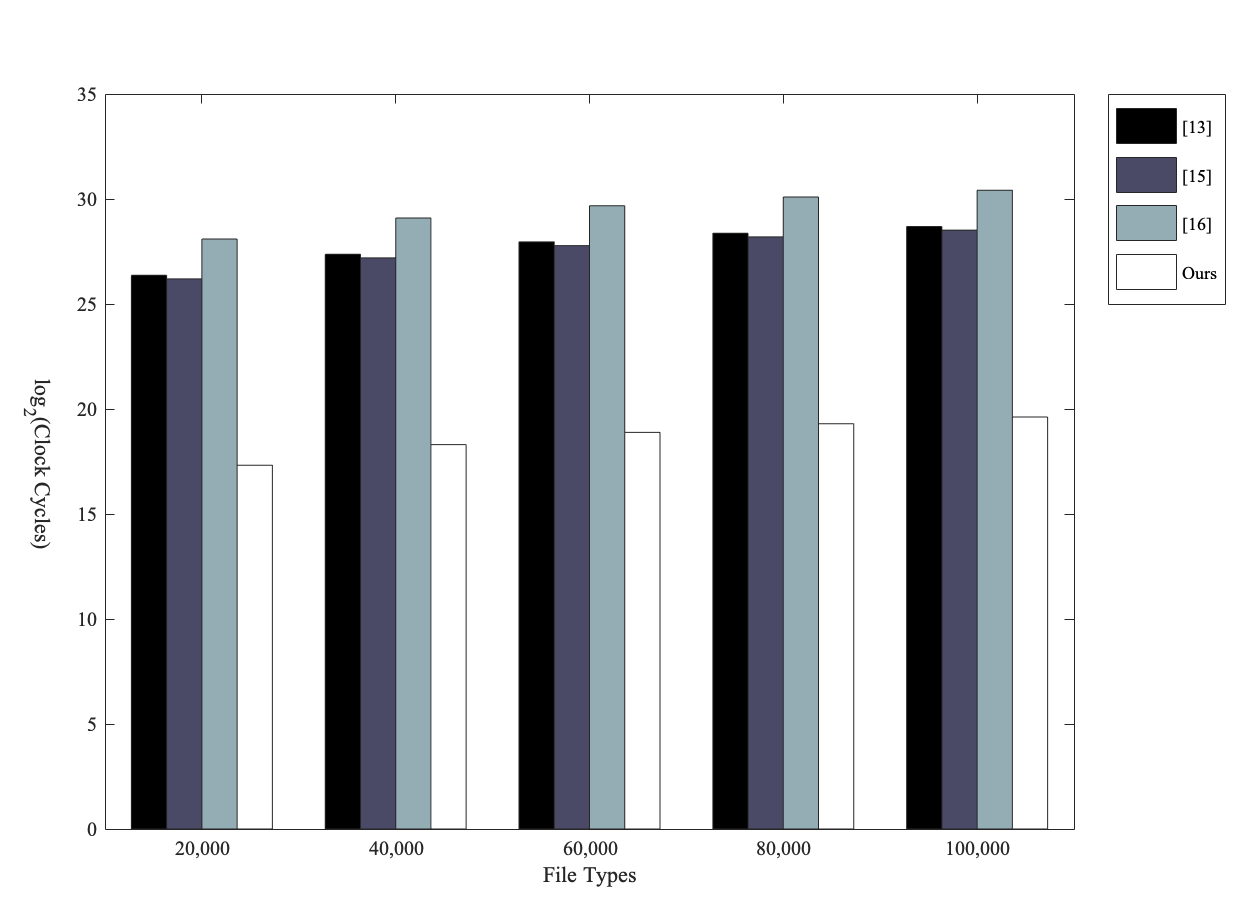}	
		\caption{Comparison of computation time of $ReKenGen$}
		\label{fig:rekeygencompare}
	\end{figure}
	
	\begin{figure}[h]
		\centering
		\includegraphics[width=0.5\textwidth]{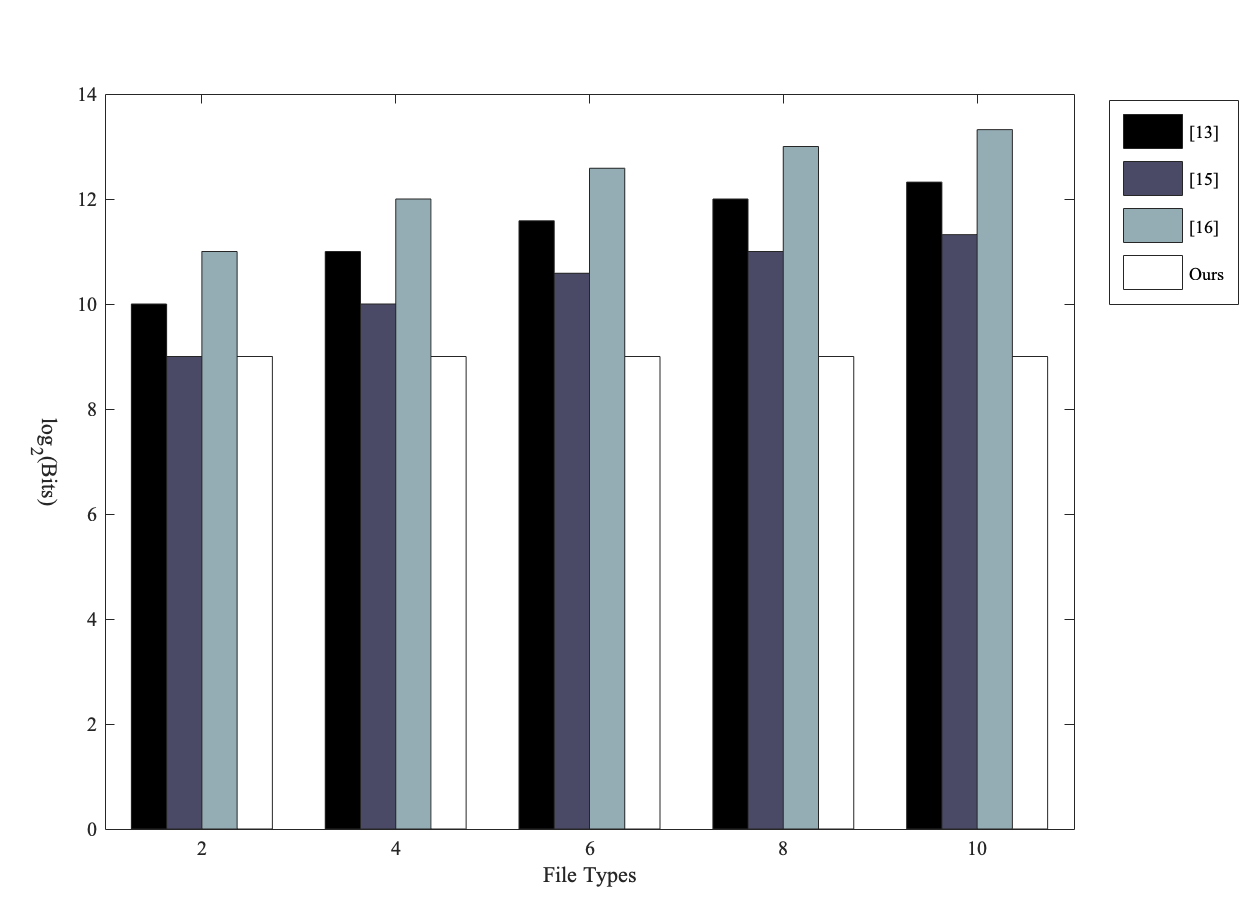}	
		\caption{Comparison of length of re-encryption keys.}
		\label{fig:rekeylengthcompare}
	\end{figure}
	\section{Conclusion}
	Cloud computing provides the benefits of fast access to data and rapid deployment by businesses or users, allowing people to access data on the cloud anytime, anywhere.
    Users do not need to do maintenance for cloud storage or pay attention to version upgrades at any time, and do not need to deal with other management issues. These features in the cloud greatly reduce the threshold for companies or users to store large amounts of data.
    This is one of the reasons why cloud computing is becoming more and more popular in various application environments. However, there are many security and privacy considerations that block the development of cloud computing.
    For example, users who value data privacy will first encrypt the data and then upload the encrypted file to the cloud. In this way, even if the files on the cloud are stolen,
    the hacker is also unable to decrypt the stolen data. Still, this method will cause great difficulties in cloud data sharing.
    For instance, Alice wants to share encrypted data with Bob, she must give her private key to Bob. It is not safe or practical to share a private key in real-world applications.
    Fortunately, Alice can adopt proxy re-encryption schemes to create a re-encryption key and then pass the re-encryption key to the cloud manager. When Bob wants to access Alice's data, the cloud manager uses Alice's re-encryption key to convert Alice's ciphertext into Bob's ciphertext.
    Unfortunately, PRE has a collusion problem. If the cloud administrator colludes with Bob, they can re-encrypt all of Alice's ciphertext into Bob's ciphertext, and use Bob's private key to decrypt Alice's ciphertext.
    In order to solve this problem, we can use C-PREs. In C-PREs, the users can manage the access control right of the ciphertext according to the condition. Therefore, the problem of collusion can be solved.
    However, we found that in CPRE, the number of re-encryption keys is proportional to the number of conditions, which is a huge burden for devices with small storage capacity.
    So we proposed KAPRE to reduce the number of re-encryption keys to a constant size. And in this paper, we prove that our proposed scheme is CCA secure, in addition to this, we also
    use the PBC library to implement our scheme and analyze it, it proves that our proposed scheme is not only theoretically safe but it can actually be applied to the real cloud applications.
	
	\label{ch:Conclusion}

	\ifCLASSOPTIONcompsoc
	
	\section*{Acknowledgments}
	\else
	regular IEEE prefers the form
	\section*{Acknowledgment}
	\fi
	
	This work was partially supported by Taiwan Information Security 
	Center at National Sun Yat-sen University (TWISC@NSYSU) and the 
	Ministry of Science and Technology of Taiwan under grant MOST 
	107-2218-E-110-014. It also was financially supported by the 
	Information Security Research Center at National Sun Yat-sen 
	University in Taiwan and the Intelligent Electronic Commerce 
	Research Center from The Featured Areas Research Center Program 
	within the framework of the Higher Education Sprout Project 
	by the Ministry of Education (MOE) in Taiwan. 

	\ifCLASSOPTIONcaptionsoff
	\newpage
	\fi

	\bibliographystyle{IEEEtran}

	\bibliography{main}

	\begin{IEEEbiography}[{\includegraphics[width=1in,height=1.25in,clip,keepaspectratio]{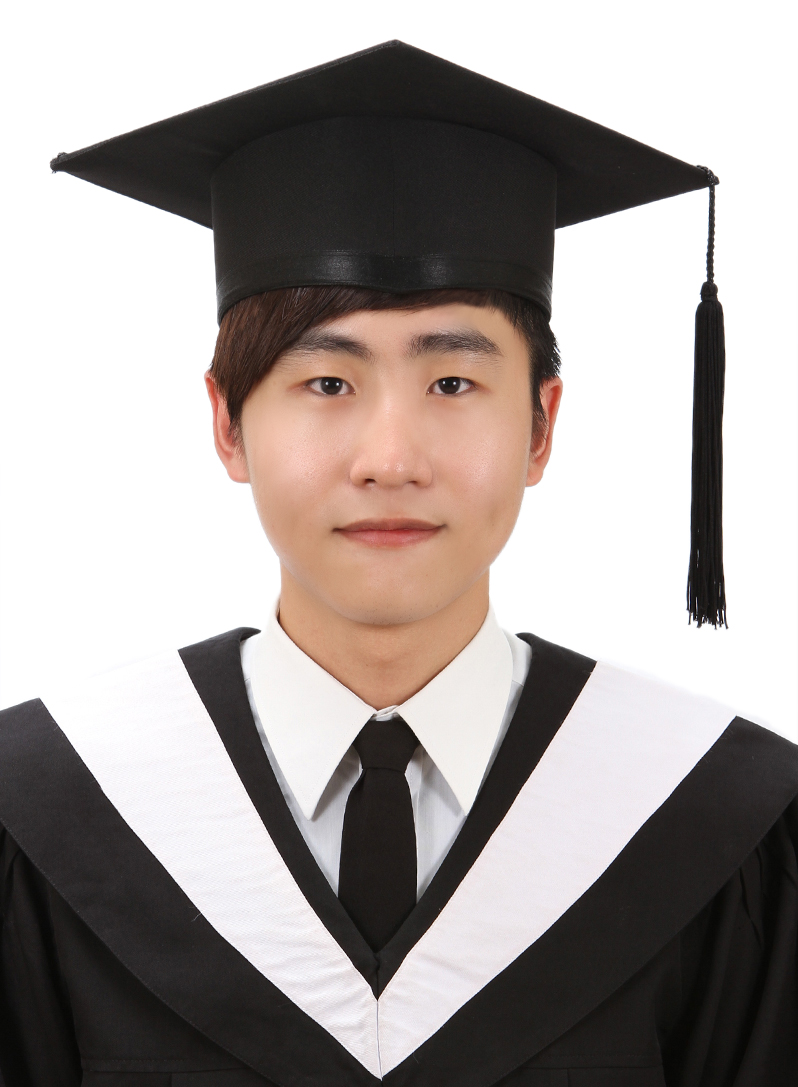}}]{Wei-Hao Chen} received his MS degree in computer science and engineering from National Sun Yat-sen University, Kaohsiung, Taiwan. His research interests include cloud computing and cloud storage, network and communication security, and applied cryptography. After completing his MS degree, he is currently a PhD student in the Department of Computer Science at Purdue University.
	\end{IEEEbiography}
	
	\begin{IEEEbiography}
		[{\includegraphics[width=1in,height=1.25in,clip,keepaspectratio]{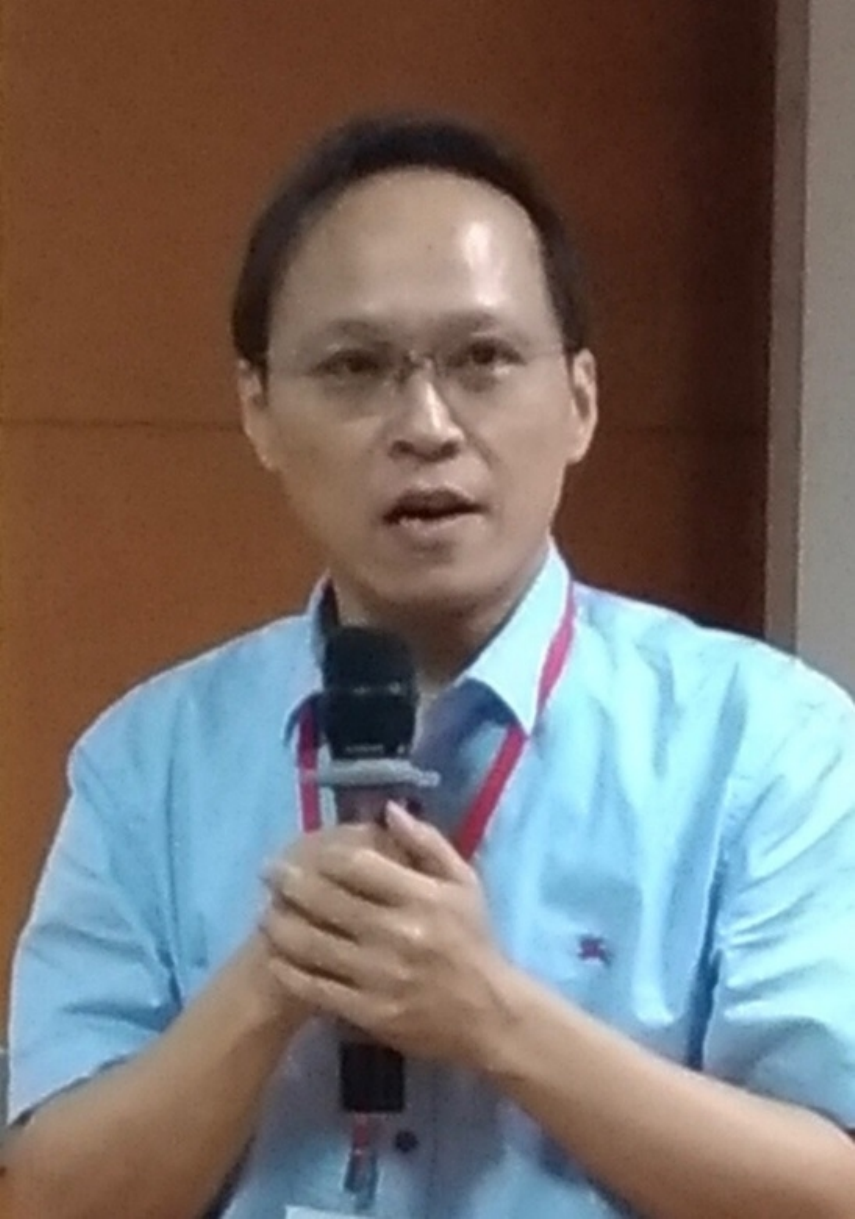}}]{Chun-I Fan}
		received the M.S. degree in computer science and information engineering from the National Chiao Tung University, Hsinchu, Taiwan, in 1993, and the Ph.D. degree in electrical engineering from the National Taiwan University, Taipei, Taiwan, in 1998. From 1999 to 2003, he was an Associate Researcher and a Project Leader with Telecommunication Laboratories, Chunghwa Telecom Company, Ltd., Taoyuan, Taiwan.  In 2003, he joined the faculty of the Department of Computer Science and Engineering, National Sun Yat-sen University, Kaohsiung, Taiwan. He has been a Full Professor since 2010 and a Distinguished Professor since 2019.  His current research interests include applied cryptology, cryptographic protocols, and information and communication security.
		Prof. Fan is the Chairman of Chinese Cryptology and Information Security Association (CCISA), Taiwan, and the Chief Executive Officer (CEO) of Telecom Technology Center (TTC), Taiwan. He also is the Director of Information Security Research Center and was the CEO of "Aim for the Top University Plan" Office at National Sun Yat-sen University.  He was the recipient of the Best Student Paper Awards from the National Conference on Information Security in 1998, the Dragon Ph.D. Thesis Award from Acer Foundation, and the Best Ph.D. Thesis Award from the Institute of Information and Computing Machinery in 1999. Prof. Fan won the Engineering Professors Award from Chinese Institute of Engineers - Kaohsiung Chapter in 2016 and he is also an Outstanding Faculty in Academic Research in National Sun Yat-sen University.
	\end{IEEEbiography}

	\begin{IEEEbiography}[{\includegraphics[width=1in,height=1.25in,clip,keepaspectratio]{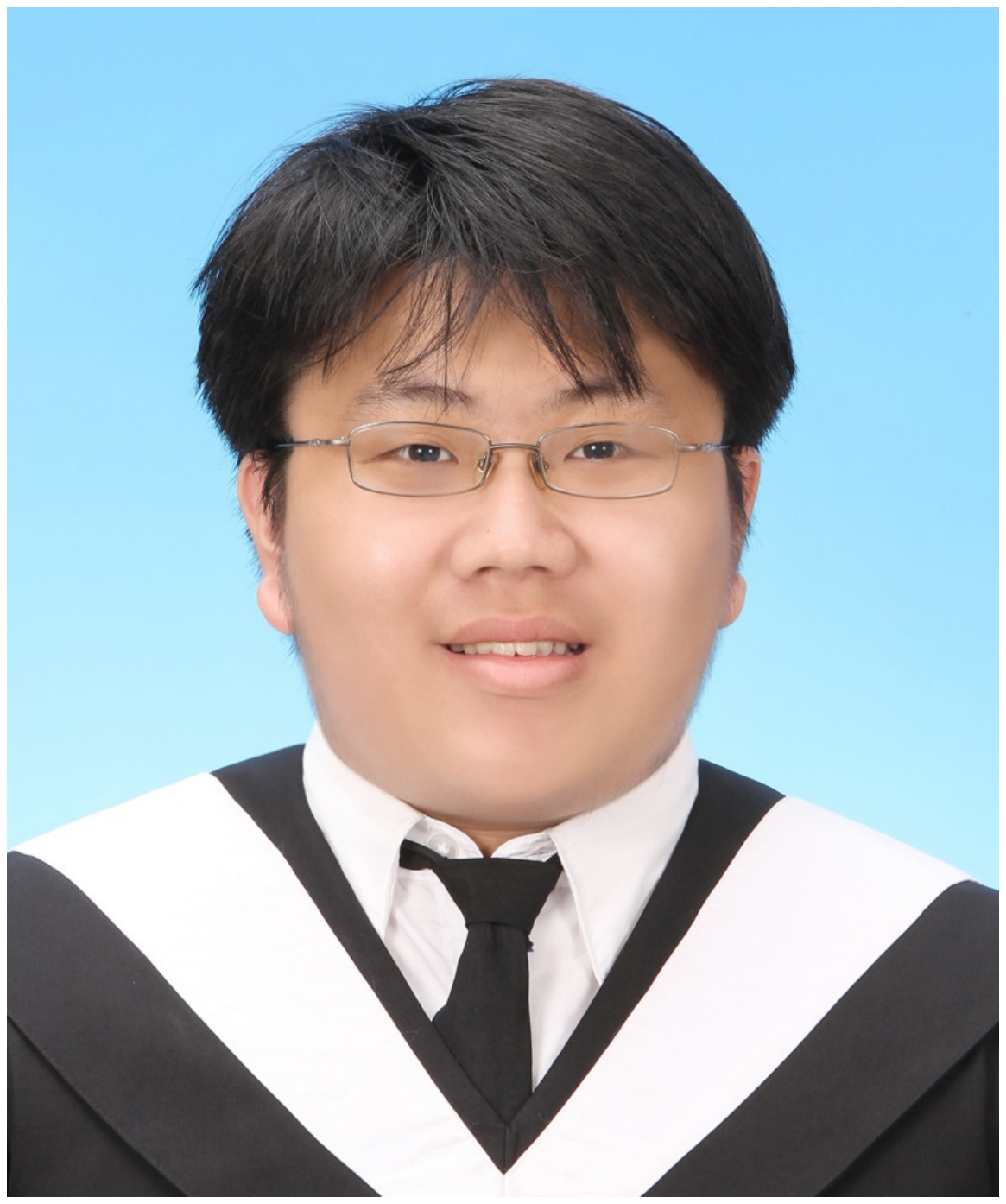}}]{Yi-Fan Tseng}
		was born in Kaohsiung, Taiwan. He received the MS degree and Ph.D. degree  in computer science and engineering from National Sun Yat-sen University, Taiwan, in 2014 and 2018, respectively. From 2018 to 2019, as a postdoctoral researcher, he joined the research group of Taiwan Information Security Center at National Sun Yat-sen University (TWISC@NSYSU). In 2019, he has joined the faculty of the Department of Computer Science, National Chengchi University, Taipei, Taiwan. His research interests
include cloud computing and security, network and communication security, information security, cryptographic protocols, and applied cryptography.
	\end{IEEEbiography}

\end{document}